\documentclass[10pt]{article}
\usepackage{bm,amsfonts,amssymb,amsthm,amsmath}
\usepackage{mathrsfs}
\usepackage{graphicx,color}
\usepackage{CJK}
\usepackage{indentfirst,enumerate}

\newcommand{\md}{\mathrm{d}}

\DeclareMathAlphabet{\mathpzc}{OT1}{pzc}{m}{it}

\newtheorem{theorem}{Theorem}
\numberwithin{theorem}{section}
\newtheorem{lemma}[theorem]{Lemma}
\newtheorem{proposition}[theorem]{Proposition}

\newtheorem{definition}[theorem]{Definition}
\numberwithin{equation}{section}

\definecolor{orange}{rgb}{1,0.5,0}
\definecolor{rb}{rgb}{1,0,1}
\allowdisplaybreaks[4]

\begin{document}

\title{\textbf{Classifying local anisotropy formed by rigid molecules: symmetries and tensors}\footnote{This work is partially supported by NCMIS, AMSS grants for Outstanding Youth Project, and ICMSEC dirctor funds.}}
\author{Jie Xu\footnote{LSEC \& NCMIS, Institute of Computational Mathematics and Scientific/Engineering Computing (ICMSEC), Academy of Mathematics and Systems Science (AMSS), Chinese Academy of Sciences, Beijing 100190, China. Email: xujie@lsec.cc.ac.cn}}
\date{}
\maketitle
\begin{abstract}
  We consider an infinitesimal volume where there are many rigid molecules of the same kind, and discuss the description and classification of the local anisotropy in this volume by tensors. 
  First, we examine the symmetry of a rigid molecule, which is described by a point group in $SO(3)$. 
  For each point group in $SO(3)$, we find the tensors invariant under the rotations in the group. 
  These tensors shall be symmetric and traceless. 
  We write down the explicit expressions. 
  The order parameters to describe the local anisotropy are then chosen as some of the invariant tensors averaged about the density function. 
  Next, we discuss the classification of local anisotropy by the symmetry of the whole infinitesimal volume. 
  This mesoscopic symmetry can be recognized by the value of the order parameter tensors in the sense of maximum entropy state. 
  For some sets of order parameter tensors involving different molecular symmetries, we give the classification of mesoscopic symmetries, in which the three-fold, four-fold and polyhedral symmetries are examined. 
  

\end{abstract}


\section{Introduction}





Let us consider a volume in mesoscopic scale, which is infinitesimal so that it can be viewed as a point in the space, while is very large compared with molecules so that it is able to contain a huge number of molecules. 
Suppose that in this infinitesimal volume, there are many identical, fully rigid molecules with different orientations. How to describe the state in this infinitesimal volume? 
The physical background of the above problem is the description of local anisotropy in liquid crystals. 
Liquid crystals are featured by partial order that results in intermediate physical properties between liquids and solids. 
The mechanism of partial order in liquid crystals is exactly local anisotropy that orginates from the building blocks, typically non-spherical molecules with rigidity. 
With nonuniform orientation distributions, ordered phases can form even in spatially homogeneous cases. 
Thus, for any mathematical theory of liquid crystals, a basic problem is how to describe and classify local anisotropy. 

Let us explain this basic problem by 
rod-like molecules. 
The orientation of one molecule can be represented by a unit vector $\bm{m}\in S^2$ fixed on it. 
The state in an infinitesimal volume can then be described by a density function $\rho(\bm{m})$. 
In most cases, the second order tensor, 
\begin{equation}
  Q=\int_{S^2}(\bm{m}\otimes\bm{m}-\frac{1}{3}\mathfrak{i})\rho(\bm{m})\,\md\bm{m}, 
\end{equation}
is defined as the order parameter, where $\mathfrak{i}$ denotes the identity matrix. 
The local anisotropy is classified by the eigenvalues of $Q$ into three cases: isotropic, uniaxial and biaxial. 
Once this tensor has been chosen as the order parameter, 
one can construct the free energy as a functional of the tensor field $Q(\bm{x})$, such as the widely-used Landau-de Gennes theory \cite{GennesPierreGillesde1993Tpol}. 
The Oseen-Frank theory \cite{Oseen_Frank,ericksen1991liquid} can also be viewed as its simplification that the local anisotropy is uniaxial everywhere. 

Liquid crystals can also be formed by other rigid molecules, such as bent-core molecules \cite{JJAP,prl_111_067801,Ntb,Ntb2} and kite-like molecules \cite{Zhao29092015} that have proved experimentally to exhibit many fascinating structures, where the local anisotropy becomes more complicated. 
On the theoretical aspect, it requires to reconsider the choice of order parameters. 
For example, the proposed order parameters for bent-core molecules may include 
one first order tensor, up to two second order tensors \cite{Bi1,bisi2006universal,de2008landau,luckhurst2012molecular,shamid2014predicting,longa2016modulated,SymmO,BentModel}, and even a third order tensor \cite{lubensky2002theory}. 
On the other hand, the local anisotropy shown by these molecules is believed to be diverse. 
The types of local anisotropy other than the uniaxial and biaxial, including polar, three-fold, four-fold \cite{JJAP,PhysRevE.58.5873}, tetrahedral \cite{fel1995tetrahedral,longa2009chiral,brand2010macroscopic,brand2005tetrahedratic,trojanowski2012tetrahedratic}, and octahedral \cite{PhysRevE.58.5873}, have been considered. 

The works mentioned above, however, focus only on some special classes of molecules. 
We consider in this paper two fundamental problems. 
For general rigid molecules, what is the principle of choosing order parameters? 
Then, based on the order parameters, how to classify the local anisotropy? 

For the first problem, it shall be pointed out that the choice of order parameters is based on molecular symmetry. 
For rod-like molecules, the tensor $Q$ is defined from the second order tensor generated by $\bm{m}$, which is the lowest order nonvanishing tensor when averaged over the density function. 
Actually, it is suitable for any molecule with the same symmetry as a rod, no matter it is a disc, a spheroid or an hourglass. 
On the other hand, as we can see from the theories for bent-core and other molecules, if the molecular symmetry changes, the nonvanishing tensors will be different. 
Therefore, we shall classify rigid molecules according to their symmetries, and discuss the nonvanishing tensors for each symmetry. 

For general rigid molecules, the orientational density function $\rho$ is a function on $SO(3)$. 
The situation becomes more complicated, because there are multiple linearly dependent tensors of the same order. 
On identifying the independent components, it will become clear that we shall focus on the symmetric traceless tensors. 
Thus, we will start from writing down explicit expressions of two bases for symmetric traceless tensors. 
Symmetric traceless tensors are related to group representation theory \cite{SpecFun}. They give the complete expansion of a function on $SO(3)$, as their components are linearly equivalent to Wigner D-functions, for which the approximation theory is established \cite{wigner_SXZ}. 
Actually, in some works \cite{PhysRevE.58.5873,doi:10.1063/1.1649733}, the coefficients in the expansion by Wigner D-functions act as order parameters. 

With the preparation above, we are ready to discuss molecular symmetries, for which we focus on proper rotations. 
For a certain molecule, all the proper rotations leaving it invariant form a closed subgroup of $SO(3)$, or a point group in $SO(3)$. 
The elements in the point group also define rotations on tensors. 
We will show that if a tensor is nonvanishing when averaged, it shall be invariant under any rotation in the point group. 
A molecule may also be invariant under improper rotations, but they will not affect invariant tensors, so we do not consider them in the current work. 
The point groups in $SO(3)$ have been completely identified (see, for example, \cite{Group_Cotton}). 
We write down the invariant tensors for \emph{all} point groups in $SO(3)$, using the explicit expressions of symmetric traceless tensors. 
The order parameters are then chosen as some of the invariant tensors averaged over the density function. 
For particular molecules, the choice of order parameters might depend on many aspects \cite{mettout2006macroscopic}. 
However, the choice should be able to distinguish different groups, for which we claim some conditions. 



We now turn to the second problem, to classify local anisotropy by chosen order parameter tensors. 
Recall that for rod-like molecules with the order parameter tensor $Q$, the isotropic, uniaxial and biaxial states have different symmetries. 
The word 'symmetry' here refers to the mesoscopic symmetry, including all the rotations of the whole infinitesimal volume leaving it unchanged. 
Thus, the classification of local anisotropy is also according to rotation point groups. 

When the infinitesimal volume is filled with rod-like molecules, the mesoscopic symmetry could actually be recognized from the tensor $Q$ by finding the maximum entropy state, i.e. the density $\rho$ that maximizes the entropy with the value of $Q$ fixed. 
Such a viewpoint is appropriate for the general case.
Although the set of tensors chosen as order parameters could be various, we could always define the mesoscopic symmetry by that of the maximum entropy state. 
Thus, we could discuss mescoscopic symmetry for any choice of tensors order parameters, for which we find out the equivalent conditions on the value of these tensors. 
The classification for rod-like molecules could be recovered from this general setting. 
Furthermore, for the same molecular symmetry, we could discuss the classification when choosing different sets of tensor order parameters, which could lead to different level of classification on mesoscopic symmetry. 
Thus, it becomes clear that choosing order parameters is also relevant to our demand on classifying local anisotropy. 
We will discuss some cases including the possibilities of forming local anisotropy of three-fold, four-fold and polyhedral symmetries, some of which have been indicated by some experiments \cite{JJAP}. 

The identification of invariant tensors, and the classification of mesoscopic symmetries, can be applied directly to the interpretation of the results from molecular/Monte Carlo simulations \cite{greco2014molecular,greco2015entropy,tomczyk2016twist, ferrarini2017twist,camp1999theory,memmer2002liquid,lansac2003phase,dewar2005dipolar,nguyen2008molecular}. 
Our results also give a preliminary for the construction of tensor model for different molecular symmetries. 
As done in some recent works \cite{SymmO,RodModel,BentModel,doi:10.1080/02678292.2017.1290285}, theories based on tensors can be derived from microscopic theory by expansion, where the molecular symmetry and truncation determine the order parameters. 
The recognition of invariant tensors is crucial for writing down the expansion, which will be discussed in a forthcoming work.

We will follow the route below in the rest of the paper. 
Some notations on $SO(3)$ and tensors are introduced in Section \ref{notation}. 
In Section \ref{sym_tt}, we discuss symmetric traceless tensors, where we write down explicit expressions of two bases and illustrate briefly the relation to Wigner D-functions. 
In Section \ref{mol_sym}, we turn to molecular symmetries. For each rotation point group, we write down the invariant tensors. 
The classification of mesoscopic symmetries for local anisotropy is discussed in Section \ref{dist_sym}. 
Concluding remarks are given in Section \ref{concl}. 

\section{Notations\label{notation}}
\subsection{Orthonormal frame}
We define a right-handed reference orthonormal frame, $(\bm{e}_1,\bm{e}_2,\bm{e}_3)$, in the space. 
For a rigid molecule, we fix another right-handed orthonormal frame $(\bm{m}_1,\bm{m}_2,\bm{m}_3)$ on it to express its orientation. 
The coordinates of the body-fixed frame in the reference frame, 
\begin{equation}
  \mathfrak{p}_{ij}=\bm{e}_i\cdot\bm{m}_j,\ i,j=1,2,3, 
\end{equation}
define a $3\times 3$ rotation matrix $\mathfrak{p}\in SO(3)$. 
The $i$-th column of $\mathfrak{p}$ is the vector $\bm{m}_i$, so we would also use $\mathfrak{p}$ to represent the body-fixed frame itself. 
Under this notation, the reference frame corresponds to the identity matrix $\mathfrak{i}$. 
The matrix can also be expressed by Euler angles, $\alpha$, $\beta$ and $\gamma$, 
\begin{align}
\mathfrak{p}=\left(
\begin{array}{ccc}
 \cos\alpha &\quad -\sin\alpha\cos\gamma &\quad\sin\alpha\sin\gamma\\
 \sin\alpha\cos\beta &\quad\cos\alpha\cos\beta\cos\gamma-\sin\beta\sin\gamma &
 \quad -\cos\alpha\cos\beta\sin\gamma-\sin\beta\cos\gamma\\
 \sin\alpha\sin\beta &\quad\cos\alpha\sin\beta\cos\gamma+\cos\beta\sin\gamma &
 \quad -\cos\alpha\sin\beta\sin\gamma+\cos\beta\cos\gamma
\end{array}
\right),\label{EulerRep}
\end{align}
where $0\le\alpha\le \pi$, $0\le\beta,\gamma <2\pi$. 

The uniform probability measure is given by $\md \mathfrak{p}=(1/8\pi^2)\sin\alpha\md\alpha\md\beta\md\gamma$. 
This measure is invariant under the rotation in $SO(3)$: for any function $f(\mathfrak{p})$, we have 
\begin{equation}\label{invariantP}
  \int f(\mathfrak{p})\md \mathfrak{p}=\int f(\mathfrak{pt})\md \mathfrak{p}=\int f(\mathfrak{tp})\md \mathfrak{p}, \quad\forall \mathfrak{t}\in SO(3). 
\end{equation}
The product $\mathfrak{pt}$ follows the rule of matrix multiplication.

Throughout the paper, the rotation $\mathfrak{p}$ is taken as a variable on which the density $\rho$ depends. 
From this viewpoint, the axes of the body-fixed frame, $\bm{m}_i$, are also regarded as functions of the variable $\mathfrak{p}$. 
Thus, if we write an expression about $\bm{m}_i$, the expression is viewed as a function of the variable $\mathfrak{p}$. 


We mention that another way to describe rotations is to use unit quaternions. In appendix, we write down the relation between quaternion and rotation matrix in $SO(3)$. 

\subsection{Tensors}
Let us introduce the notations involving tensors in the three-dimensional space. 
An $n$-th order tensor $U$ is an element in the space $\underbrace{\mathbb{R}^3\otimes\ldots\otimes \mathbb{R}^3}_{n\text{ times}}$, 
of which a basis 
can be given by 
\begin{equation}\label{basis0}
\bm{e}_{i_1}\otimes\ldots\otimes\bm{e}_{i_n}, \quad i_1,\ldots,i_n\in\{1,2,3\}. 
\end{equation}
We can write an $n$-th order tensor $U$ as a linear combination of the basis, 
\begin{equation}\label{expnd0}
  U=U_{i_1\ldots i_n}\bm{e}_{i_1}\otimes\ldots\otimes\bm{e}_{i_n}, 
\end{equation}
where the multidimensional array $U_{i_1\ldots i_n}$ is called the coordinates, or components, of the tensor $U$ under the basis \eqref{basis0}. 
In the above, we assume the Einstein summation convention on repeated indices that will be used throughout the paper.

An important tensor is the second order identity tensor.
If written in matrix, it coincides with the identity $\mathfrak{i}$ in $SO(3)$, so we also use the notation $\mathfrak{i}$ for the identity tensor. 

The dot product can be defined for two tensors of the same order. 
If $U_1$ and $U_2$ are both $n$-th order tensors, their dot product is defined as summing up the product of the corresponding coordinates under an orthonormal basis, 
\begin{equation}
  U_1\cdot U_2=(U_1)_{i_1\ldots i_n}(U_2)_{i_1\ldots i_n}. 
\end{equation}
Note that the definition is independent of the orthonormal basis we choose. 
The coordinates of the tensor can then be expressed as the dot product with the basis \eqref{basis0}, 
\begin{equation}
  U_{i_1\ldots i_n}=U\cdot\bm{e}_{i_1}\otimes\ldots\otimes\bm{e}_{i_n}. 
\end{equation}

A tensor $U$ is called a symmetric tensor, if for arbitrary $i_{\sigma_1}$ and $i_{\sigma_2}$, it satisfies 
  $U_{\ldots i_{\sigma_1}\ldots i_{\sigma_2}\ldots}=U_{\ldots i_{\sigma_2}\ldots i_{\sigma_1}\ldots}$. 
For a symmetric tensor, we can define its trace as the contraction of two indices. 
The trace transforms an $n$-th order symmetric tensor to an $(n-2)$-th order symmetric tensor: 
\begin{equation}\label{symtrlsdef}
  (\text{tr}U)_{i_1\ldots i_{n-2}}=U_{i_1\ldots i_{n-2}jj}. 
\end{equation}
If a symmetric tensor $U$ satisfies $\mathrm{tr}U=0$, we say that $U$ is symmetric traceless. 
We shall note that concepts of symmetric tensor and symmetric traceless tensor are also independent of the choice of orthonormal basis, although they are defined through the components. 

For a rotation $\mathfrak{p}\in SO(3)$, let us define how it is acted on a tensor. 
By the definition, the rotation of the frame $(\bm{e}_i)$ is given by $\mathfrak{p}\circ\bm{e}_i=\mathfrak{p}_{ji}\bm{e}_j=\bm{m}_i$. 
Thus, for any tensor $U$ written in the form \eqref{expnd0}, we could define the rotation by $\mathfrak{p}\in SO(3)$ on the tensor as follows, 
\begin{align}
  \mathfrak{p}\circ U=&U_{i_1\ldots i_n}\bm{m}_{i_1}\otimes\ldots\otimes\bm{m}_{i_n}
  .\label{rot0}
\end{align}
Thus, for any tensor $U$, we can view $U(\mathfrak{p})=\mathfrak{p}\circ U$ as a function of $\mathfrak{p}$. 
The tensor that is not rotated can also be viewed as rotated by the identity, which can be written as $U=U(\mathfrak{i})$. 

One can verify that the rotation does not rely on the choice of basis. 
Moreover, we can deduce that $U(\mathfrak{p}_1\mathfrak{p}_2)=\mathfrak{p}_1\circ U(\mathfrak{p}_2)$. 
As a result, it can be verified that for the dot product, 
\begin{equation}
  U_1(\mathfrak{sp}_1)\cdot U_2(\mathfrak{sp}_2)=U_1(\mathfrak{p}_1)\cdot U_2(\mathfrak{p}_2), \quad\forall \mathfrak{s},\mathfrak{p}_1,\mathfrak{p}_2\in SO(3). 
\end{equation}
Besides, the rotation keeps symmetric and symmetric traceless properties, 
because the original tensor and the rotated tensor have the same coordinates under different bases. 

\section{Symmetric traceless tensors\label{sym_tt}}
We would like to describe the anisotropy in a small volume containing many identical non-spherical rigid molecules. 
The anisotropy is generated by non-uniform orientational distribution of these rigid molecules. 
Recall that we fix an orthonormal frame $\mathfrak{p}=(\bm{m}_1,\bm{m}_2,\bm{m}_3)$ on each rigid molecule to represent the orientation. 
Therefore, the orientational distribution can be expressed by a density function $\rho(\mathfrak{p})$ in $SO(3)$. 
%
As a natural extension of the $Q$-tensor, we consider the moments of $\bm{m}_i$, i.e. the averages of the tensor products of several $\bm{m}_i$,
\begin{equation}
\langle \bm{m}_{i_1}\otimes\ldots\otimes\bm{m}_{i_n}\rangle = \int_{SO(3)}\bm{m}_{i_1}(\mathfrak{p})\otimes\ldots\otimes\bm{m}_{i_n}(\mathfrak{p})\rho(\mathfrak{p})\md \mathfrak{p},\quad i_1\ldots,i_n=1,2,3.
\label{moment0}
\end{equation}
Hereafter, we use $\langle\ldots\rangle$ to denote the average of a function on $SO(3)$ about $\rho(\mathfrak{p})$. 
However, these moments, or equivalently the integrands, are not linearly independent. 
The rest of this section is dedicated to figuring out the linearly independent quantities in these moments.
The results turn out to be fundamental when we discuss molecular symmetries in the next section. 

First, we notice from the cross product relation $\bm{m}_1\times\bm{m}_2=\bm{m}_3$ that the components of $\bm{m}_1\otimes\bm{m}_2-\bm{m}_2\otimes \bm{m}_1$ are given by those of $\bm{m}_3$. 
This guides us to focus on symmetric tensors. 
For a $n$-th order tensor $U$, we define its permutational average as $U_{\mathrm{sym}}$, 
\begin{equation}\label{permavg}
(U_{\mathrm{sym}})_{i_1\ldots i_n}=\frac{1}{n!}\sum_{\sigma}U_{i_{\sigma(1)}\ldots i_{\sigma(n)}},
\end{equation}
where the sum is taken over all the permutations $\sigma$ of $\{1,\ldots,n\}$.
It is clear that $U_{\mathrm{sym}}$ is a symmetric tensor. 
Then, the tensor $U$ is decomposed into the symmetric part $U_{\mathrm{sym}}$ and the antisymmetric part $U-U_{\mathrm{sym}}$. 
For the antisymmetric part of the tensor $\bm{m}_{i_1}\otimes\ldots\otimes\bm{m}_{i_n}$, the components are given by the components of some $(n-1)$-th order tensors.
One can repeat such a decomposition for the resulting $(n-1)$-th order tensors. 
Hence, when seeking linearly independent components, we only need to consider symmetric tensors. 

To express the symmetric tensors conveniently, we introduce the monomial notation, 
\begin{align}
  \bm{m}_1^{k_1}\bm{m}_2^{k_2}\bm{m}_3^{k_3}
  =(\underbrace{\bm{m}_1\otimes\ldots\otimes\bm{m}_1}_{k_1}\otimes
  \underbrace{\bm{m}_2\otimes\ldots\otimes\bm{m}_2}_{k_2}\otimes
  \underbrace{\bm{m}_3\otimes\ldots\otimes\bm{m}_3}_{k_3}
  )_{\mathrm{sym}}.
\label{monomial0}
\end{align}
It is then straightforward to regard a polynomial of $\bm{m}_i$ as a symmetric tensor, if every term has the same order. 
Note that for any orthonormal frame $(\bm{m}_1,\bm{m}_2,\bm{m}_3)$, we have 
\begin{equation}
\mathfrak{i}=\bm{m}_1^2+\bm{m}_2^2+\bm{m}_3^2, \label{i_m}
\end{equation}
where $\mathfrak{i}$ is the identity matrix. 
So, the identity matrix $\mathfrak{i}$ can be regarded as a polynomial of $\bm{m}_i$. 
We also define the product of $\mathfrak{i}^l$ and a symmetric tensor $U$ as 
\begin{equation}
  \mathfrak{i}^lU=(\underbrace{\mathfrak{i}\otimes\ldots \otimes\mathfrak{i}}_{l}\otimes U)_{\mathrm{sym}}. 
\end{equation}

For the integrand in \eqref{moment0}, we only need to look at its symmetric part $\bm{m}_1^{k_1}\bm{m}_2^{k_2}\bm{m}_3^{k_3}$ to find out linearly independent tensors. 
However, there are still some linear relations by noticing \eqref{i_m}. 
If $2k\le k_3\le 2k+1$, we could write 
\begin{equation}\label{m3}
\bm{m}_1^{k_1}\bm{m}_2^{k_2}\bm{m}_3^{k_3}=\bm{m}_1^{k_1}\bm{m}_2^{k_2}(\mathfrak{i}-\bm{m}_1^2-\bm{m}_2^2)^k\bm{m}_3^{k_3-2k}. 
\end{equation}
Since the coordinates of $\mathfrak{i}^lU$ are actually given by those of $U$, 
\eqref{m3} actually gives linear relations between the coordinates of tensors of different orders. 
%

It turns out that we could arrive at clearer relations by investigating symmetric traceless tensors.
Obviously, the symmetric traceless tensors of a certain order $n$ form a linear space. 
For this linear space, we write down two bases explicitly, 
one constructed from monomials, followed by an orthonogonal basis. 
%
%
%
%
The explicit expressions will be given, which are essential to our discussion on symmetries afterwards. 
When discussing bases, we could use any orthonormal frame in $\mathbb{R}^3$. 
So, instead of using the frame $(\bm{m}_i)$ in the above, we will discuss in the frame $(\bm{e}_i)$. 
In general, by using $(\bm{m}_i)$ we are emphasizing quantities as functions of $\mathfrak{p}$, and by using $(\bm{e}_i)$ we are not emphasizing this.

\subsection{Basis constructed from monomials}

Let us begin with a lemma explaining how we calculate the trace. 
\begin{lemma}\label{trI}
Suppose $V$ is an $(n-2k)$-th order symmetric tensor.  Then we have 
\begin{equation}\label{traceI}
  \mathrm{tr}(\mathfrak{i}^kV)=\frac{2k(2n-2k+1)}{n(n-1)}\mathfrak{i}^{k-1}V+\frac{(n-2k)(n-2k-1)}{n(n-1)}\mathfrak{i}^k\mathrm{tr}V. 
\end{equation}
Here, we take $\mathrm{tr}V=0$ if $V$ is a zeroth or first order tensor. 
\end{lemma}

The lemma can be established by expanding the tensor using \eqref{permavg}. 
Similar to the above lemma, we can deduce that 
\begin{align}
  \mathrm{tr}(\bm{e}_1^{k_1}\bm{e}_2^{k_2}\bm{e}_3^{k_3})=&\frac{1}{(k_1+k_2+k_3)(k_1+k_2+k_3-1)}\Big(k_1(k_1-1)\bm{e}_1^{k_1-2}\bm{e}_2^{k_2}\bm{e}_3^{k_3}\nonumber\\
  &+k_2(k_2-1)\bm{e}_1^{k_1}\bm{e}_2^{k_2-2}\bm{e}_3^{k_3}+k_3(k_3-1)\bm{e}_1^{k_1}\bm{e}_2^{k_2}\bm{e}_3^{k_3-2}\Big). \label{tracem}
\end{align}

With the above formulae, we derive the expression of tensors generated by monomials $\bm{e}_1^{k_1}\bm{e}_2^{k_2}\bm{e}_3^{k_3}$ for $k_3=0,1$. 
The equations \eqref{traceI} and \eqref{tracem} prompt us to consider the tensor of the following form, 
\begin{equation}
(\bm{e}_1^{k_1}\bm{e}_2^{k_2}\bm{e}_3^{k_3})_0=\sum_{2j_1\le k_1,2j_2\le k_2} a^{k_1,k_2,k_3}_{j_1,j_2}(\bm{e}_1^{k_1-2j_1}\bm{e}_2^{k_2-2j_2}\bm{e}_3^{k_3}\mathfrak{i}^{j_1+j_2}),\quad k_3=0,1. \label{SymTrls_monomial}
\end{equation}
Let $n=k_1+k_2+k_3$.
We calculate the trace of this tensor and let it be zero, leading to the recursive formula, 
\begin{align}
&a^{k_1,k_2,k_3}_{j_1-1,j_2}(k_1-2j_1+1)(k_1-2j_1+2)
+a^{k_1,k_2,k_3}_{j_1,j_2-1}(k_2-2j_2+1)(k_2-2j_2+2)\nonumber\\
&~~+a^{k_1,k_2,k_3}_{j_1,j_2}2(j_1+j_2)(2n+1-2j_1-2j_2)=0. \nonumber
\end{align}
Let $a^{k_1,k_2,k_3}_{0,0}=1$. We solve that
\begin{equation}
a^{k_1,k_2,k_3}_{j_1,j_2}=(-1)^{j_1+j_2}{j_1+j_2 \choose j_1}
\frac{k_1!k_2!(2n-1-2j_1-2j_2)!!}
{(k_1-2j_1)!(k_2-2j_2)!(2n-1)!!(2j_1+2j_2)!!}.  \label{TrZeroCoef}
\end{equation}

For $k_3\ge 2$, let $j$ be the integer such that $2j\le k_3\le 2j+1$. 
The symmetric traceless tensors are defined by
\begin{equation}
  (\bm{e}_1^{k_1}\bm{e}_2^{k_2}\bm{e}_3^{k_3})_0=(\bm{e}_1^{k_1}\bm{e}_2^{k_2}(-\bm{e}_1^2-\bm{e}_2^2)^j\bm{e}_3^{k_3-2j})_0. \label{SymTrls_monomial2}
\end{equation}
By this definition, we shall see that the above tensor can also be written in the form $\bm{e}_1^{k_1}\bm{e}_2^{k_2}\bm{e}_3^{k_3}-\mathfrak{i}V$, because we have 
\begin{align*}
(\bm{e}_1^{k_1}\bm{e}_2^{k_2}(-\bm{e}_1^2-\bm{e}_2^2)^j\bm{e}_3^{k_3-2j})_0
=&\bm{e}_1^{k_1}\bm{e}_2^{k_2}(-\bm{e}_1^2-\bm{e}_2^2)^j\bm{e}_3^{k_3-2j}-\mathfrak{i}V_1\\
=&(\bm{e}_1^{k_1}\bm{e}_2^{k_2}(\mathfrak{i}-\bm{e}_1^2-\bm{e}_2^2)^j\bm{e}_3^{k_3-2j}-\mathfrak{i}V_2)-\mathfrak{i}V_1\\ 
=&\bm{e}_1^{k_1}\bm{e}_2^{k_2}\bm{e}_3^{k_3}-\mathfrak{i}(V_1+V_2). 
\end{align*}

One could notice that we can also derive symmetric traceless tensors by constraining $k_1=0,1$ or $k_2=0,1$, following the same route above. 
We shall point out in the following proposition that this leads to the same results. 
\begin{proposition}\label{st_unique}
  For any $n$-th order symmetric tensor $U$, suppose an $(n-2)$-th order symmetric tensor $V$ makes $U-\mathfrak{i}V$ traceless, then $V$ is unique. 
\end{proposition}
\begin{proof}
  Suppose there are two such tensors $V_1\ne V_2$. Then, we deduce that $\mathfrak{i}(V_1-V_2)$ is traceless. We could write $V_1-V_2=\mathfrak{i}^{k-1}W$ where $W\ne 0$ cannot be written as the product of $\mathfrak{i}$ and another tensor. 

  However, the above assumption on $W$ leads to contradiction. 
  Indeed, using \eqref{traceI}, we deduce that 
  $$
  \mathrm{tr}(\mathfrak{i}^kW)=\mathfrak{i}^{k-1}\left(\frac{2k(2n-2k+1)}{n(n-1)}W+\frac{(n-2k)(n-2k-1)}{n(n-1)}\mathfrak{i}\mathrm{tr}W\right)=0, 
  $$
  indicating that $W=-\frac{(n-2k)(n-2k-1)}{2k(2n-2k+1)}\mathfrak{i}\mathrm{tr}W$. 
\end{proof}
Thus, the above derivation actually gives a way to construct symmetric traceless tensors from any symmetric tensor $U$ with the form $U-\mathfrak{i}V$. 
\begin{definition}\label{u0}
For any symmetric tensor $U$, we define $(U)_0$ as the symmetric traceless tensor constructed by $U-\mathfrak{i}W$. It can be done by expressing $U$ as a linear combination of $\bm{e}_1^{k_1}\bm{e}_2^{k_2}\bm{e}_3^{k_3}$. 
\end{definition}

\begin{proposition}
The $2n+1$ symmetric traceless tensors given by \eqref{SymTrls_monomial} and \eqref{TrZeroCoef} are linearly independent. Thus, they give a basis for $n$-th order symmetric traceless tensors, whose dimension is $2n+1$. 
\end{proposition}
\begin{proof}
Suppose some constants $a_{(k_1k_2k_3)}$ make 
$$
\sum_{\substack{k_1+k_2+k_3=n\\k_3=0,1}}a_{(k_1k_2k_3)}(\bm{e}_1^{k_1}\bm{e}_2^{k_2}\bm{e}_3^{k_3})_0=0. 
$$
According to Definition \ref{u0},
there exists a tensor $W$ such that 
$$
\sum_{\substack{k_1+k_2+k_3=n\\k_3=0,1}}a_{(k_1k_2k_3)}\bm{e}_1^{k_1}\bm{e}_2^{k_2}\bm{e}_3^{k_3}=\mathfrak{i}W. 
$$
Let us expand the above in the form of \eqref{expnd0}. On the left-hand side, only the components with at most one $\bm{e}_3$ can be nonzero, so $W$ can only take zero. 
It means that the left-hand side is also zero, leading to $a_{(k_1k_2k_3)}=0$. 
\end{proof}

\subsection{Orthogonal basis}
The construction of orthogonal basis is closely related to group representation \cite{SpecFun}. 

We denote by $\sqrt{-1}$ the imaginary unit. 
First, we consider $(\bm{e}_2+\sqrt{-1}\bm{e}_3)^k$. 
We can calculate directly that 
$$
\mathrm{tr}(\bm{e}_2+\sqrt{-1}\bm{e}_3)^2=\mathrm{tr}(\bm{e}_2^2-\bm{e}_3^2+2\sqrt{-1}\bm{e}_2\bm{e}_3)=0. 
$$
Therefore, $(\bm{e}_2+\sqrt{-1}\bm{e}_3)^k$ is a symmetric traceless tensor. 
We expand this tensor and look at its real and imaginary parts, which are both symmetric traceless. 
In particular, we would like to express the two tensors by $\bm{e}_1^{k_1}\bm{e}_2^{k_2}\bm{e}_3^{k_3}$ with $k_3=0,1$. 
We shall see that 
\begin{align}
  &(\bm{e}_2+\sqrt{-1}\bm{e}_3)^n=\sum_{2k\le n}(-1)^k{n\choose 2k}\bm{e}_2^{n-2k}\bm{e}_3^{2k}+\sqrt{-1}\sum_{2k+1\le n}(-1)^k{n\choose 2k+1}\bm{e}_2^{n-2k-1}\bm{e}_3^{2k+1}\nonumber\\
  =&\sum_{2k\le n}(-1)^k{n\choose 2k}\bm{e}_2^{n-2k}(\mathfrak{i}-\bm{e}_1^2-\bm{e}_2^2)^{k}
  +\sqrt{-1}\sum_{2k+1\le n}(-1)^k{n\choose 2k+1}\bm{e}_2^{n-2k-1}(\mathfrak{i}-\bm{e}_1^2-\bm{e}_2^2)^k\bm{e}_3\nonumber\\
  =&\tilde{T}_n(\bm{e}_2,\mathfrak{i}-\bm{e}_1^2)+\sqrt{-1}\tilde{U}_{n-1}(\bm{e}_2,\mathfrak{i}-\bm{e}_1^2)\bm{e}_3.
  \label{chebyshev0}
\end{align}
We could check the coefficients and find that $\tilde{T}_n$ and $\tilde{U}_{n-1}$ are defined from Chebyshev polynomials, 
$$
\tilde{T}_n(y,z)=z^{n/2}T_n(y/\sqrt{z}),\quad \tilde{U}_n(y,z)=z^{n/2}U_n(y/\sqrt{z}), 
$$ 
where $T_n(\cos\theta)=\cos n\theta$ and $U_n(\cos\theta)\sin\theta=\sin(n+1)\theta$ are Chebyshev polynomials of the first and the second kind, respectively. 

Based on the above two tensors, we can derive an orthogonal basis of $n$-th order symmetric traceless tensors. 
Consider the tensors in the following form, 
\begin{equation}
  \bigg(a_{k,n}\bm{e}_1^k+\sum_{j=1}^{[k/2]}a_{j;k,n}\bm{e}_1^{k-2j}\mathfrak{i}^j\bigg)V,
\end{equation}
where $V=\tilde{T}_{n-k}(\bm{e}_2,\mathfrak{i}-\bm{e}_1^2)$ or $\tilde{U}_{n-k-1}(\bm{e}_2,\mathfrak{i}-\bm{e}_1^2)\bm{e}_3$. 
We require the above tensor to be traceless. 
Similar to the derivation of $(\bm{e}_1^{k_1}\bm{e}_2^{k_2}\bm{e}_3^{k_3})_0$, 
by calculating the trace using \eqref{traceI}, we solve that 
\begin{equation}
  a_{j;k,n}=(-1)^j\frac{k!(2n-2j-1)!!}{2^j j! (k-2j)!(2n-1)!!}a_{k,n}.
\end{equation}
We recognize that the coefficients are proportional to those of the Jacobi polynomials with two identical indices. 
Define $\tilde{P}_k^{(\mu,\mu)}(y,z)=z^{k/2}P_k^{(\mu,\mu)}(y/\sqrt{z})$,
where $P_k^{(\mu,\mu)}(x)$ is the Jacobi polynomial with the indices $(\mu,\mu)$.
The tensors derived above give $2n+1$ $n$-th order symmetric traceless tensors, 
\begin{align}
  &
  \tilde{P}_k^{(n-k,n-k)}(\bm{e}_1,\mathfrak{i})\tilde{T}_{n-k}(\bm{e}_2,\mathfrak{i}-\bm{e}_1^2),\quad k=0,\ldots,n,\nonumber\\
  &
  \tilde{P}_k^{(n-k,n-k)}(\bm{e}_1,\mathfrak{i})\tilde{U}_{n-k-1}(\bm{e}_2,\mathfrak{i}-\bm{e}_1^2)\bm{e}_3,\quad k=0,\ldots,n-1.  \label{SymTrls_gamma}
\end{align}
In the special case $k=n$, we obtain the tensor $(\bm{e}_1^n)_0$ (cf. \eqref{TrZeroCoef}). 

\begin{proposition}
The $2n+1$ tensors given in \eqref{SymTrls_gamma} form an orthogonal basis of $n$-th order symmetric traceless tensors. 
\end{proposition}
\begin{proof}
Because we have known the dimension is $2n+1$, we only need to show the orthogonality. 
Assume $k>k'$. 
We calculate the dot product of the following two tensors, 
$$
\Big(\tilde{P}_k^{(n-k,n-k)}(\bm{e}_1,\mathfrak{i})\tilde{T}_{n-k}(\bm{e}_2,\mathfrak{i}-\bm{e}_1^2)\Big)\cdot \Big(\tilde{P}_{k'}^{(n-k',n-k')}(\bm{e}_1,\mathfrak{i})\tilde{T}_{n-k'}(\bm{e}_2,\mathfrak{i}-\bm{e}_1^2)\Big). 
$$
Using the fact that the second tensor is traceless, we could eliminate all the $\mathfrak{i}$ in $\tilde{P}_k^{(n-k,n-k)}$, so that it equals to 
$$
\Big(b_k^{(n-k,n-k)}\bm{e}_1^k\tilde{T}_{n-k}(\bm{e}_2,\mathfrak{i}-\bm{e}_1^2)\Big)\cdot \Big(\tilde{P}_{k'}^{(n-k',n-k')}(\bm{e}_1,\mathfrak{i})\tilde{T}_{n-k'}(\bm{e}_2,\mathfrak{i}-\bm{e}_1^2)\Big), 
$$
where $b_k^{(n-k,n-k)}$ is the leading coefficient of $\tilde{P}_k^{(n-k,n-k)}$. 
Note that $\tilde{T}_{n-k'}$ can be written as a linear combination of $\bm{e}_2^{n-k'-2j}\bm{e}_3^{2j}$, so the second tensor can be written as linear combination of $\bm{e}_1^{k_1}\bm{e}_2^{k_2}\bm{e}_3^{k_3}$ with $k_1\le k'<k$. 
Thus, certain coordinate of the two tensors cannot be nonzero simultaneously, therefore the dot product is zero.
The orthogonality of other pairs of tensors can be shown similarly. 
\end{proof}

The above derivation is suitable for the rotation of the symmetric traceless tensors by any $\mathfrak{p}\in SO(3)$, i.e. $\mathfrak{p}\circ (\bm{e}_1^{k_1}\bm{e}_2^{k_2}\bm{e}_3^{k_3})_0=(\bm{m}_1^{k_1}\bm{m}_2^{k_2}\bm{m}_3^{k_3})_0$. 
We will see shortly that the two bases of symmetric traceless tensors derived in this section are convenient for us to derive nonvanishing tensors for each symmetry. 

\subsection{Components of symmetric traceless tensors and Wigner D-functions}
For the space of $n$-th order symmetric traceless tensors, let us choose any basis $X^n_1,\ldots, X^n_{2n+1}$, where the superscript $n$ represents the tensor order. 
For each $X^n_{i}(\mathfrak{p})$ as a function of $\mathfrak{p}$, we could express it by the basis $X^n_{j}=X^n_{j}(\mathfrak{i})$. So, the components of $X^n_{i}(\mathfrak{p})$ give at most $2n+1$ linearly independent scalar functions of $\mathfrak{p}$, for which we can choose $X^n_{j}(\mathfrak{i})\cdot X^n_{i}(\mathfrak{p})$ for $1\le j\le 2n+1$. Since $i$ also ranges from $1$ to $2n+1$, we have written down $(2n+1)^2$ scalar functions of $\mathfrak{p}$. 
Define $\mathbb{T}^n$ as the function space formed by the components of the tensors $U(\mathfrak{p})$ where the order of $U$ is not greater than $n$. 
Besides, we define $\mathbb{T}_{\mathrm{sym},0}^n$ as the function space spanned by the components of symmetric traceless tensors $\big(U(\mathfrak{p})\big)_0$ of the order $n$, or equivalently those spanned by $X^n_{j}(\mathfrak{i})\cdot X^n_{i}(\mathfrak{p})$. 
By our discussion above, we have 
\begin{align}
  &\mathbb{T}^n=\sum_{j=0}^n\mathbb{T}_{\mathrm{sym},0}^j,\quad
  \mathrm{dim}\mathbb{T}_{\mathrm{sym},0}^n\le (2n+1)^2. \label{tsum}
\end{align}

We point out that the functions in $\mathbb{T}_{\mathrm{sym},0}^n$ are linearly equivalent to Wigner D-functions with the major index $n$, so that they give a complete expansion for functions on $SO(3)$. 
This statement is also related to group representation, 
which we do not attempt to introduce here. 
Instead, we will write down a quick derivation to verify directly. 
To introduce Wigner D-functions, we state them as eigenfunctions of differential operators in $SO(3)$.
%
Denote by $L_i$ the derivatives along the infinitesimal rotation about $\bm{m}_i$, which satisfy (see, for example, \cite{wigner_SXZ})
\begin{align}
  L_1\bm{m}_{2}=\bm{m}_{3},\ L_2\bm{m}_{3}=\bm{m}_{1},\ L_3\bm{m}_{1}=\bm{m}_{2},\ L_1\bm{m}_{3}=-\bm{m}_{2},\ L_2\bm{m}_{1}=-\bm{m}_{3},\ L_3\bm{m}_{2}=-\bm{m}_{1}. 
  \label{diffL}
\end{align}
The Wigner D-functions can be written by Euler angles as 
\begin{equation}
  D_{mm'}^n\big(\mathfrak{p}\big)=\exp(-\sqrt{-1}m\beta)d_{mm'}^n(\alpha)\exp(-\sqrt{-1}m'\gamma),\quad m,m'=-n,-n+1,\ldots,n. 
\end{equation}
The function $d_{mm'}^n(\alpha)$, is a trigonometric polynomial given in the following form, 
\begin{equation}
  d_{mm'}^n(\alpha)=c_{mm'}^n\left(\sin\frac{\alpha}{2}\right)^a
  \left(\cos\frac{\alpha}{2}\right)^b
  P_k^{(a,b)}(\cos\alpha), 
\end{equation}
where $k=n-\max (|m|,|m'|)$, $a=|m-m'|$, $b=|m+m'|$, $c_{mm'}^n$ are some constants, 
and $P_k^{(a,b)}$ is the Jacobi polynomial. 
When we fix $n$, the functions $D_{mm'}^n$ give linearly independent eigenfunctions of the Laplacian $L^2=L_1^2+L_2^2+L_3^2$, a self-adjoint operator, with the eigenvalue $-n(n+1)$. 

\begin{proposition}\label{eigen}
The functions in $\mathbb{T}_{\mathrm{sym},0}^n$ are eigenfunctions of $L^2$ with the eigenvalue $-n(n+1)$. 
\end{proposition}
\begin{proof}
We calculate the Laplacian of monomials using \eqref{diffL}, 
\begin{align}
&-L^2(\bm{m}_1^{k_1}\bm{m}_2^{k_2}\bm{m}_3^{k_3}\mathfrak{i}^{n-k_1-k_2-k_3})\nonumber\\
=&n(n+1)\bm{m}_1^{k_1}\bm{m}_2^{k_2}\bm{m}_3^{k_3}\mathfrak{i}^{n-k_1-k_2-k_3}-k_1(k_1-1)\bm{m}_1^{k_1-2}\bm{m}_2^{k_2}\bm{m}_3^{k_3}\mathfrak{i}^{n-k_1-k_2-k_3+2}\nonumber\\
&-k_2(k_2-1)\bm{m}_1^{k_1}\bm{m}_2^{k_2-2}\bm{m}_3^{k_3}\mathfrak{i}^{n-k_1-k_2-k_3+2}-k_3(k_3-1)\bm{m}_1^{k_1}\bm{m}_2^{k_2}\bm{m}_3^{k_3-2}\mathfrak{i}^{n-k_1-k_2-k_3+2}. 
\end{align}
Take it into \eqref{SymTrls_monomial} to obtain the result. 
\end{proof}

\begin{proposition}\label{wt}
  The function space spanned by $D_{mm'}^n$ with fixed $n$ satisfies 
  \begin{equation}
    \mathrm{span}\{D_{mm'}^n\}= \mathbb{T}_{\mathrm{sym},0}^n. 
  \end{equation}
\end{proposition}
\begin{proof}
  First, we shall show
  \begin{equation}
    D_{mm'}^n\in \mathbb{T}_{\mathrm{sym},0}^n. 
  \end{equation}
  Let us express $D_{mm'}^n$ by the polynomials of $\mathfrak{p}_{ij}$. 
  By the Euler angles representation, we have
  \begin{align*}
    &2v_{1,\pm}=(\mathfrak{p}_{22}+\mathfrak{p}_{33})\pm \sqrt{-1}(\mathfrak{p}_{32}-\mathfrak{p}_{23})=(1+\cos\alpha)\exp\big(\pm\sqrt{-1}(\beta+\gamma)\big),\\ 
    &2v_{2,\pm}=(\mathfrak{p}_{33}-\mathfrak{p}_{22})\mp \sqrt{-1}(\mathfrak{p}_{23}+\mathfrak{p}_{32})=(1-\cos\alpha)\exp\big(\pm\sqrt{-1}(\beta-\gamma)\big),
  \end{align*}
  Without loss of generality, we assume $m\ge |m'|\ge 0$. 
  \begin{enumerate}
  \item If $m-m'$ is even, we have 
    \begin{align*}
      D_{mm'}^n=&c_{mm'}^n\exp(-\sqrt{-1}m\beta)\exp(-\sqrt{-1}m'\gamma)
      \left(\sin\frac{\alpha}{2}\right)^{m-m'}
      \left(\cos\frac{\alpha}{2}\right)^{m+m'}
      P_k^{(a,b)}(\cos\alpha)\\
      =&c_{mm'}^n
      v_{1,-}^{\frac{m-m'}{2}}
      v_{2,-}^{\frac{m+m'}{2}}
      P_k^{(a,b)}(\mathfrak{p}_{11}). 
    \end{align*}
  \item If $m-m'$ is odd, we have $m-1\ge|m'|$. Thus, 
    \begin{align*}
      D_{mm'}^n=&c_{mm'}^n\exp(-\sqrt{-1}(m-1)\beta)\exp(-\sqrt{-1}m'\gamma)
      \left(\sin\frac{\alpha}{2}\right)^{m-m'-1}
      \left(\cos\frac{\alpha}{2}\right)^{m+m'-1}\\
      &\exp(-\sqrt{-1}\beta)\sin\frac{\alpha}{2}\cos\frac{\alpha}{2}
      P_k^{(a,b)}(\cos\alpha)\\
      =&c_{mm'}^n
      v_{1,-}^{\frac{m-m'-1}{2}}
      v_{2,-}^{\frac{m+m'-1}{2}}
      \frac{\mathfrak{p}_{21}-\sqrt{-1}\mathfrak{p}_{31}}{2}
      P_k^{(a,b)}(\mathfrak{p}_{11}). 
    \end{align*}
  \end{enumerate}
  So we have 
  \begin{equation}
    D_{mm'}^n\in \mathbb{T}^n=\sum_{j=1}^n\mathbb{T}_{\mathrm{sym},0}^j, 
  \end{equation}
  where we use \eqref{tsum}.
  Proposition \ref{eigen} implies that the sum in the above is a direct sum.
  Since $D_{mm'}^n$ is an eigenfunction of $L^2$ with the eigenvalue $-n(n+1)$, we obtain $D_{mm'}^n\in \mathbb{T}_{\mathrm{sym},0}^n$.

  On the other hand, the dimension of the set $\mathrm{span}\{D_{mm'}^n\}$ equals to $(2n+1)^2$. Together with the inequality in \eqref{tsum}, we conclude the proof. 
\end{proof}

\section{Molecular symmetries and invariant tensors\label{mol_sym}}
The molecular symmetry is characterized by the rotations leaving the molecule invariant. 
When looking at the body-fixed frame $\mathfrak{p}$, such a rotation is given by $\mathfrak{s}\in SO(3)$ that transforms $\mathfrak{p}$ to $\mathfrak{ps}$. 
To comprehend this transformation, one could pose the body-fixed frame coincide with the reference frame $(\bm{e}_i)$, then rotate it by $\mathfrak{s}$, followed by $\mathfrak{p}$, resulting in $\mathfrak{ps}$. 

Let us denote by $\mathcal{G}$ the set of all such $\mathfrak{s}$ leaving the molecule invariant.
It is easy to see that
$\mathcal{G}$ is a subgroup of $SO(3)$.
Since for $\mathfrak{s}\in\mathcal{G}$, the orientation $\mathfrak{ps}$ does not differ from $\mathfrak{p}$. 
Therefore, the density function satisfies
\begin{equation}
  \rho(\mathfrak{ps})=\rho(\mathfrak{p}). \label{molsym}
\end{equation}
It shall be clarified that a rigid molecule might also allow improper rotations.
This is, however, beyond the scope of this paper, since an improper rotation brings the body-fixed frame into a left-handed one. 
But the density function $\rho(\mathfrak{p})$, about which the tensors are averaged, is defined on $SO(3)$. 
Therefore, the improper rotations cannot be reflected in $\rho(\mathfrak{p})$. 
The improper rotations will play a role when considering interaction between rigid molecules. 

When symmetric traceless tensors are averaged about $\rho(\mathfrak{p})$, \eqref{molsym} gives some relations for them. 
To simplify the presentation, let us assume $\mathcal{G}$ is finite and illustrate the relations. 
For any tensor $T(\mathfrak{p})$, using \eqref{invariantP}, its average equals to 
\begin{align}
  \langle T(\mathfrak{p})\rangle=\int T(\mathfrak{p})\rho(\mathfrak{p})\md \mathfrak{p}=&\frac{1}{\#\mathcal{G}}\int \sum_{\mathfrak{s}\in \mathcal{G}}T(\mathfrak{p})\rho(\mathfrak{ps})\md \mathfrak{p} 
  =\frac{1}{\#\mathcal{G}}\int \sum_{\mathfrak{s}\in \mathcal{G}}T(\mathfrak{ps})\rho(\mathfrak{p})\md \mathfrak{p}. \label{tg0}
\end{align}
Denote by $T^{\mathcal{G}}$ the average of $T(\mathfrak{p})$ over ${\mathcal{G}}$, 
\begin{equation}
T^{\mathcal{G}}(\mathfrak{p})=\frac{1}{\#\mathcal{G}}\sum_{\mathfrak{s}\in \mathcal{G}}T(\mathfrak{ps}). \label{tg}
\end{equation}
Then, \eqref{tg0} indicates that 
$
\langle T(\mathfrak{p})\rangle=\langle T^{\mathcal{G}}(\mathfrak{p})\rangle. 
$
Therefore, when considering the moments, we only need to focus on the set 
$$
\mathbb{A}^{\mathcal{G},n}=\{T^{\mathcal{G}}(\mathfrak{p}): T \text{ is an $n$-th order symmetric traceless tensor}\}. 
$$

\begin{proposition}\label{inv_van}
  The space of $n$-th order symmetric traceless tensors can be decomposed as the direct sum of two linear subspaces that are mutually orthogonal. 
  One consists of all the tensors that are invariant under $\mathcal{G}$, 
  \begin{equation}
    \mathbb{A}^{\mathcal{G},n}=\{T(\mathfrak{p}): T(\mathfrak{ps})=T(\mathfrak{p}),\ \forall \mathfrak{s}\in \mathcal{G}\}.\label{tgset}
  \end{equation}
  The other consists of tensors that are vanishing when averaged over $\mathcal{G}$, 
  \begin{equation}
    (\mathbb{A}^{\mathcal{G},n})^{\perp}=\{T(\mathfrak{p}): T^{\mathcal{G}}(\mathfrak{p})=0\}. \label{tgcompset}
  \end{equation}
\end{proposition}
\begin{proof}
Using the fact that $\mathcal{G}$ is a group, it is noted from \eqref{tg} that if $\mathfrak{s}\in\mathcal{G}$, the expression for $T^{\mathcal{G}}(\mathfrak{ps})$ and $T^{\mathcal{G}}(\mathfrak{p})$ are the same. So they are equal, indicating that $T^{\mathcal{G}}(\mathfrak{p})$ belongs to the set given in the right-hand side of \eqref{tgset}. 
On the other hand, if $T(\mathfrak{ps})=T(\mathfrak{p})$ for any $\mathfrak{s}\in\mathcal{G}$, then by \eqref{tg} we have $T^{\mathcal{G}}(\mathfrak{p})=T(\mathfrak{p})$, leading to $T(\mathfrak{p})\in \mathbb{A}^{\mathcal{G},n}$. 
Thus, we have \eqref{tgset}. 

Let us consider the orthogonal complement $(\mathbb{A}^{\mathcal{G},n})^{\perp}$. 
We shall show that $U^{\mathcal{G}}(\mathfrak{p})=0$ is equivalent to $U\in (\mathbb{A}^{\mathcal{G},n})^{\perp}$. 
Indeed, for any $T\in \mathbb{A}^{\mathcal{G},n}$, 
\begin{align*}
  T(\mathfrak{p})\cdot U^{\mathcal{G}}(\mathfrak{p})=&\frac{1}{\#\mathcal{G}}\sum_{\mathfrak{s}\in \mathcal{G}}T(\mathfrak{p})\cdot U(\mathfrak{ps})
  =\frac{1}{\#\mathcal{G}}\sum_{\mathfrak{s}\in \mathcal{G}}T(\mathfrak{ps}^{-1}\mathfrak{p}^{-1}\mathfrak{p})\cdot U(\mathfrak{ps}^{-1}\mathfrak{p}^{-1}\mathfrak{ps})\\
  =&\frac{1}{\#\mathcal{G}}\sum_{\mathfrak{s}\in \mathcal{G}}T(\mathfrak{ps}^{-1})\cdot U(\mathfrak{p})
  =\frac{1}{\#\mathcal{G}}\sum_{\mathfrak{s}\in \mathcal{G}}T(\mathfrak{ps})\cdot U(\mathfrak{p})\\
  =&T^{\mathcal{G}}(\mathfrak{p})\cdot U(\mathfrak{p})=T(\mathfrak{p})\cdot U(\mathfrak{p}). 
\end{align*}
So, $T(\mathfrak{p})\cdot U^{\mathcal{G}}(\mathfrak{p})=0$ is equivalent to $T(\mathfrak{p})\cdot U(\mathfrak{p})=0$. 
On the other hand, we have $U^{\mathcal{G}}\in \mathbb{A}^{\mathcal{G}}$. 
Thus, $T(\mathfrak{p})\cdot U^{\mathcal{G}}(\mathfrak{p})=0$ for any $T\in\mathbb{A}^{\mathcal{G},n}$ is equivalent to $U^{\mathcal{G}}(\mathfrak{p})=0$. 
\end{proof}

Based on the above proposition, we call those tensors in the space $\mathbb{A}^{\mathcal{G},n}$ invariant tensors of $\mathcal{G}$, and those tensors in the orthogonal complement $(\mathbb{A}^{\mathcal{G},n})^{\perp}$ vanishing tensors of $\mathcal{G}$. 
As we have mentioned, order parameters for each symmetry will be the averages of some tensors chosen from the space $\mathbb{A}^{\mathcal{G},n}$. 
Thus, our task is to write down the space $\mathbb{A}^{\mathcal{G},n}$. 

We state below a relation of the invariant tensors about subgroups that will help our discussion. 
\begin{proposition}\label{tensor_subgroup}
  If $\mathcal{H}$ is a subgroup of $\mathcal{G}$, then by invariance, $\mathbb{A}^{\mathcal{G},n}$ is a linear subspace of $\mathbb{A}^{\mathcal{H},n}$, thus $(\mathbb{A}^{\mathcal{H},n})^{\perp}$ is a linear subspace of $(\mathbb{A}^{\mathcal{G},n})^{\perp}$. 
\end{proposition}
Therefore, when discussing the invariant tensors for certain group $\mathcal{G}$, we could examine some subgroups of $\mathcal{G}$ to help. 
Moreover, when verifying invariance under $\mathcal{G}$, we just need to verify for generating elements of $\mathcal{G}$. 

Since the rotation of the rigid molecule is regarded as a continuous map, we would only consider the closed subgroups of $SO(3)$, i.e. the point groups consisting of rotations. 
The point groups in $SO(3)$ have been fully classified (see, for example, \cite{Group_Cotton}). 
They include two axisymmetric groups $\mathcal{C}_{\infty}$ and $\mathcal{D}_{\infty}$; two classes of finite axial groups $\mathcal{C}_{n}$ and $\mathcal{D}_n$; three polyhedral groups $\mathcal{T}$, $\mathcal{O}$ and $\mathcal{I}$. 
In what follows, we clarify our convention on how to put the body-fixed frame $\bm{m}_i$, explain the generating elements in the point groups, and write down the invariant tensors for these groups. 

\subsection{Axial groups}

Let us introduce two rotations.
We use $\mathfrak{j}_{\theta}$ to represent the rotation round $\bm{m}_1$ for the angle $\theta$, and $\mathfrak{b}_2$ to represent the rotation round $\bm{m}_2$ for the angle $\pi$. 
The two matrices are written as follows, 
\begin{align}
  \mathfrak{j}_{\theta}=\left(
  \begin{array}{ccc}
    1 & 0 & 0\\
    0 & \cos\theta & -\sin\theta\\
    0 & \sin\theta & \cos\theta
  \end{array}
  \right), \quad
  \mathfrak{b}_2=\left(
  \begin{array}{ccc}
    -1 & 0 & 0\\
    0 & 1 & 0\\
    0 & 0 & -1
  \end{array}
  \right). \label{basicrot}
\end{align}
To see the rotations more clearly, let us write down how the frame $\mathfrak{p}=(\bm{m}_1,\bm{m}_2,\bm{m}_3)$ is rotated by $\mathfrak{j}_{\theta}$, for which we just apply the matrix multiplication to obtain 
\begin{align*}
  \mathfrak{pj}_{\theta}=\big(\bm{m}_1(\mathfrak{pj}_{\theta}),\bm{m}_2(\mathfrak{pj}_{\theta}),\bm{m}_3(\mathfrak{pj}_{\theta})\big)=(\bm{m}_1, \cos\theta\bm{m}_2+\sin\theta\bm{m}_3, -\sin\theta\bm{m}_2+\cos\theta\bm{m}_3). 
\end{align*}
We could calculate that 
\begin{align}
  \big(\bm{m}_2(\mathfrak{pj}_{\theta})+\sqrt{-1}\bm{m}_3(\mathfrak{pj}_{\theta})\big)^k
  =&e^{-\sqrt{-1}k\theta}(\bm{m}_2+\sqrt{-1}\bm{m}_3)^k. \label{m23rotate}
\end{align}
The rotation $\mathfrak{b}_2$ is written as $\mathfrak{pb}_2=(-\bm{m}_1,\bm{m}_2,-\bm{m}_3)$. 

We start from the cyclic group $\mathcal{C}_n$.
It allows the rotation round an axis with the angle $2\pi/n$, which is also the generating element. 
We pose the body-fixed frame so that $\bm{m}_1$ is the rotational axis.
Then, the generating element is $\mathfrak{j}_{2\pi/n}$. 
The group $\mathcal{C}_n$ can be written as
$$
\mathcal{C}_n=\{\mathfrak{i},\mathfrak{j}_{2\pi/n},\mathfrak{j}_{2\pi/n}^2,\ldots,\mathfrak{j}_{2\pi/n}^{n-1}\}. 
$$
%
Consider the average of $(\bm{m}_2+\sqrt{-1}\bm{m}_3)^k$ over the cyclic group $\mathcal{C}_n$. 
If $k$ is a multiple of $n$, then by \eqref{m23rotate}, 
$$
\big(\bm{m}_2(\mathfrak{pj}_{\theta})+\sqrt{-1}\bm{m}_3(\mathfrak{pj}_{\theta})\big)^k=e^{-2\pi\sqrt{-1}k/n}(\bm{m}_2+\sqrt{-1}\bm{m}_3)^k=(\bm{m}_2+\sqrt{-1}\bm{m}_3)^k
$$
is invariant. Otherwise, we could calculate that 
$$
\Big((\bm{m}_2+\sqrt{-1}\bm{m}_3)^k\Big)^{\mathcal{C}_n}=\frac{1}{n}\sum_{j=0}^{n-1}e^{-\sqrt{-1}\cdot 2jk\pi/n}(\bm{m}_2+\sqrt{-1}\bm{m}_3)^k=0, 
$$
is vanishing. 
Besides, the rotations in the group $\mathcal{C}_n$ keep $\bm{m}_1$ invariant. 
By noticing \eqref{chebyshev0}, we have found that $\mathbb{A}^{C_n,l}$ is given by 
\begin{align}
  \mathbb{A}^{\mathcal{C}_n,l}=\mathrm{span}\Big\{\tilde{P}_{l-jn}^{(jn,jn)}(\bm{m}_1,\mathfrak{i})\tilde{T}_{jn}(\bm{m}_2,\mathfrak{i}-\bm{m}_1^2),\ 
  &\tilde{P}_{l-jn}^{(jn,jn)}(\bm{m}_1,\mathfrak{i})\tilde{U}_{jn-1}(\bm{m}_2,\mathfrak{i}-\bm{m}_1^2)\bm{m}_3,\ jn\le l\Big\}. \label{Cn}
\end{align}

Next, we discuss the dihedral group $\mathcal{D}_n$.
This group contains the cyclic group $\mathcal{C}_n$ as its subgroup. 
Besides, it contains a rotation by the angle $\pi$ round an axis perpendicular to the axis of the $n$-fold rotation.
We pose the body-fixed frame such that $\bm{m}_1$ coincides with the $n$-fold axis, and $\bm{m}_2$ is the two-fold axis. 
The two generating elements are now given by $\mathfrak{j}_{2\pi/n}$ and $\mathfrak{b}_2$. 
Since $\mathcal{D}_n$ contains $\mathcal{C}_n$ as a subgroup, $\mathbb{A}^{\mathcal{D}_n,l}\subseteq \mathbb{A}^{\mathcal{C}_n,l}$. 
Furthermore, the tensors in $\mathbb{A}^{\mathcal{D}_n,l}$ shall be invariant under the subgroup $\{\mathfrak{i},\mathfrak{b}_2\}$.
Since $\mathfrak{pb}_2=(-\bm{m}_1,\bm{m}_2,-\bm{m}_3)$, the tensor is invariant only when the order of $\bm{m}_1$ and $\bm{m}_3$ are both odd or both even. 
So, it is easy to recognize that in \eqref{Cn}, the first tensor is invariant when $l-jn$ is even, and vanishing when $l-jn$ is odd; the second tensor is invariant when $l-jn$ is odd, and vanishing when $l-jn$ is even.
Thus, 
\begin{align}
  \mathbb{A}^{D_n,l}=\mathrm{span}&\left\{\tilde{P}_{l-jn}^{(jn,jn)}(\bm{m}_1,\mathfrak{i})\tilde{T}_{jn}(\bm{m}_2,\mathfrak{i}-\bm{m}_1^2),\ l-jn\ge 0 \text{ even};\right. \nonumber\\
  &\left.\tilde{P}_{l-jn}^{(jn,jn)}(\bm{m}_1,\mathfrak{i})\tilde{U}_{jn-1}(\bm{m}_2,\mathfrak{i}-\bm{m}_1^2)\bm{m}_3,\ l-jn\ge 0 \text{ odd} \right\}. \label{Dn}
\end{align}

\begin{theorem}
The invariant tensors for $\mathcal{C}_n$, $\mathcal{D}_n$ are given by \eqref{Cn} and \eqref{Dn}, respectively. 
\end{theorem}

We then discuss the two axisymmetries, $\mathcal{C}_{\infty}$ and $\mathcal{D}_{\infty}$.
The former contains rotation with arbitrary angle round an axis.
The latter contains those rotations, as well as the two-fold rotations round any direction perpendicular to that axis. 
It is natural to put $\bm{m}_1$ as the axis.
Obviously, $\mathcal{C}_n$ is a subset of $\mathcal{C}_{\infty}$, so $\mathbb{A}^{C_{\infty},l}\subseteq \mathbb{A}^{C_n,l}$ for arbitrary $n$.
As a result, we must have 
$
\mathbb{A}^{C_{\infty},l}\subseteq \mathrm{span}\left\{\tilde{P}_{l}^{(0,0)}(\bm{m}_1,\mathfrak{i})\right\}. 
$
These tensors are all invariant under $\mathcal{C}_{\infty}$. 
Therefore, 
\begin{equation}
  \mathbb{A}^{\mathcal{C}_{\infty},l}=\mathrm{span}\left\{\tilde{P}_{l}^{(0,0)}(\bm{m}_1,\mathfrak{i})\right\}. \label{Cinf}
\end{equation}
The group $\mathcal{D}_{\infty}$ contains $\mathfrak{b}_2$, resulting in 
\begin{equation}
  \mathbb{A}^{\mathcal{D}_{\infty},l}=\mathrm{span}\left\{\tilde{P}_{l}^{(0,0)}(\bm{m}_1,\mathfrak{i}),\ l\text{ even}\right\}. \label{Dinf}
\end{equation}
\begin{theorem}
The invariant tensors for $\mathcal{C}_{\infty}$, $\mathcal{D}_{\infty}$ are \eqref{Cinf} and \eqref{Dinf}, respectively. 
\end{theorem}
Note that the Jacobi polynomials $\tilde{P}_{l}^{(0,0)}$ actually give the Legendre polynomials. 
For the group $\mathcal{D}_{\infty}$ allowed by the rod-like molecules, the space $\mathbb{A}^{\mathcal{D}_{\infty},l}$ is not a zero space only when $l$ is even.
The lowest order invariant tensor is $\tilde{P}_2^{(0,0)}(\bm{m}_1)=\frac{3}{2}\bm{m}_1^2-\frac{1}{2}\mathfrak{i}$, from which the tensor $Q$ is defined. 

\subsection{Polyhedral groups}
There are three polyhedral groups that are relevant to the rotations allowed by regular polyhedrons. 

The tetrahedral group $\mathcal{T}$ is formed by all the proper rotations of a regular tetrahedron, so let us explain by considering a regular tetrahedron.
We could choose the body-fixed frame so that the vertices of the regular tetrahedron are located at
$$
\lambda(\bm{m}_1+\bm{m}_2+\bm{m}_3),\ \lambda(\bm{m}_1-\bm{m}_2-\bm{m}_3),\ \lambda(-\bm{m}_1+\bm{m}_2-\bm{m}_3),\ \lambda(-\bm{m}_1-\bm{m}_2+\bm{m}_3), 
$$
where $\lambda$ is associated with the size of the regular tetrahedron. 
Any element in $\mathcal{T}$ shall define a permutation within these four vertices.
As it turns out, the group $\mathcal{T}$ can be generated by three rotations: 
\begin{itemize}
\item $\mathfrak{j}_{\pi}$, the rotation round $\bm{m}_1$ by the angle $\pi$. It transforms $(\bm{m}_1,\bm{m}_2,\bm{m}_3)$ into $(\bm{m}_1,-\bm{m}_2,-\bm{m}_3)$.
\item $\mathfrak{b}_2$, the rotation round $\bm{m}_2$ by the angle $\pi$. It transforms $(\bm{m}_1,\bm{m}_2,\bm{m}_3)$ into $(-\bm{m}_1,\bm{m}_2,-\bm{m}_3)$.
  Thus, $\mathfrak{j}_{\pi}\mathfrak{b}_2$ rotates the frame into $(-\bm{m}_1,-\bm{m}_2,\bm{m}_3)$. 
\item $\mathfrak{r}_3$, the three-fold rotation transforming $(\bm{m}_1,\bm{m}_2,\bm{m}_3)$ into $(\bm{m}_2,\bm{m}_3,\bm{m}_1)$. We can see that $\mathfrak{r}_3^2$ takes the frame into $(\bm{m}_3,\bm{m}_1,\bm{m}_2)$. 
\end{itemize}

For the invariant tensors, we consider the symmetric traceless tensor generated by a polynomial $h(\bm{m}_1,\bm{m}_2,\bm{m}_3)$ where all the terms have the same order, written as $\big(h(\bm{m}_1,\bm{m}_2,\bm{m}_3)\big)_0$. 
Let us first consider the subgroup $\{\mathfrak{i},\mathfrak{j}_{\pi},\mathfrak{b}_2,\mathfrak{j}_{\pi}\mathfrak{b}_2\}$.
The three rotations that are not identity take two of $\bm{m}_1$, $\bm{m}_2$, $\bm{m}_3$ to their opposites. 
Therefore, in the tensor $\big(h(\bm{m}_1,\bm{m}_2,\bm{m}_3)\big)_0$, it can only possess terms with the order of $\bm{m}_i$ all even or all odd.
Because of Definition \ref{u0}, this is equivalent to that $h(\bm{m}_1,\bm{m}_2,\bm{m}_3)$ only possesses terms with the order of $\bm{m}_i$ all even or all odd. 
So, the polynomial $h$ can be written either as $g(\bm{m}_1^2,\bm{m}_2^2,\bm{m}_3^2)$ or $\bm{m}_1\bm{m}_2\bm{m}_3g(\bm{m}_1^2,\bm{m}_2^2,\bm{m}_3^2)$ where $g$ is a polynomial. 

Next, let us examine the subgroup $\{\mathfrak{i},\mathfrak{r}_3,\mathfrak{r}_3^2\}$. 
The invariance under $\mathfrak{r}_3$ requires the polynomial $g(\bm{m}_1^2,\bm{m}_2^2,\bm{m}_3^2)$ to satisfy $g(\bm{m}_1^2,\bm{m}_2^2,\bm{m}_3^2)=g(\bm{m}_2^2,\bm{m}_3^2,\bm{m}_1^2)=g(\bm{m}_3^2,\bm{m}_1^2,\bm{m}_2^2)$. 
Such a polynomial $g$ can be further decomposed as the sum of two polynomials $g_1$ and $g_2$, such that $g_1$ and $g_2$ both meet the above requirement for $g$, as well as 
\begin{align*}
  g_1(\bm{m}_1^2,\bm{m}_2^2,\bm{m}_3^2)=g_1(\bm{m}_2^2,\bm{m}_1^2,\bm{m}_3^2),\ g_2(\bm{m}_1^2,\bm{m}_2^2,\bm{m}_3^2)=-g_2(\bm{m}_2^2,\bm{m}_1^2,\bm{m}_3^2). 
\end{align*}
For $g_1$, it could be expressed as a polynomial about $\bm{m}_1^2+\bm{m}_2^2+\bm{m}_3^2$, $\bm{m}_1^2\bm{m}_2^2+\bm{m}_2^2\bm{m}_3^2+\bm{m}_3^2\bm{m}_1^2$, $\bm{m}_1^2\bm{m}_2^2\bm{m}_3^2$. 
Suppose that the order of $g_1(\bm{m}_1^2,\bm{m}_2^2,\bm{m}_3^2)$ is $2j$.
Then, $g_1$ could be linearly expressed by the following polynomials, 
\begin{align*}
  (\bm{m}_1^2+\bm{m}_2^2+\bm{m}_3^2)^{j_0}(\bm{m}_1^2\bm{m}_2^2+\bm{m}_2^2\bm{m}_3^2+\bm{m}_3^2\bm{m}_1^2)^{j_1}(\bm{m}_1^2\bm{m}_2^2\bm{m}_3^2)^{j_2},\  j_0+2j_1+3j_2=j, 
\end{align*}
which are linearly independent. 
Because of Proposition \ref{st_unique}, when considering symmetric traceless tensors $\big(g_1(\bm{m}_1^2\bm{m}_2^2\bm{m}_3^2)\big)_0$ or $\big(\bm{m}_1\bm{m}_2\bm{m}_3g_1(\bm{m}_1^2\bm{m}_2^2\bm{m}_3^2)\big)_0$, all the terms with the factor $\bm{m}_1^2+\bm{m}_2^2+\bm{m}_3^2=\mathfrak{i}$ generate zero tensor.
So, we only need to keep the terms with $j_0=0$. 
For $g_2$, it could be written as 
\begin{align*}
  (\bm{m}_1^2-\bm{m}_2^2)(\bm{m}_2^2-\bm{m}_3^2)(\bm{m}_3^2-\bm{m}_1^2)\tilde{g}_2(\bm{m}_1^2,\bm{m}_2^2,\bm{m}_3^2), 
\end{align*}
where $\tilde{g}_2$ satisfies the same condition as $g_1$.
Summarizing the above discussion, the space $\mathbb{A}^{\mathcal{T},l}$ is given by 
%
\begin{align}
\mathbb{A}^{\mathcal{T},l}=\textrm{span}\bigg\{
&\Big((\bm{m}_1^2\bm{m}_2^2+\bm{m}_2^2\bm{m}_3^2+\bm{m}_3^2\bm{m}_1^2)^{j_1}(\bm{m}_1\bm{m}_2\bm{m}_3)^{j_2}\Big)_0,\ 4j_1+3j_2=l; \nonumber\\
&\Big((\bm{m}_1^2-\bm{m}_2^2)(\bm{m}_2^2-\bm{m}_3^2)(\bm{m}_3^2-\bm{m}_1^2)(\bm{m}_1^2\bm{m}_2^2+\bm{m}_2^2\bm{m}_3^2+\bm{m}_3^2\bm{m}_1^2)^{j_1}\nonumber\\&\qquad(\bm{m}_1\bm{m}_2\bm{m}_3)^{j_2}\Big)_0, 
\quad 6+4j_1+3j_2=l\bigg\}.\label{tensors_poly}
\end{align}

Next, we discuss the octahedral group $\mathcal{O}$, which contains all the proper rotations allowed by a cube.
It is natural to put the axes $\bm{m}_i$ of the body-fixed frame along the direction of the edges.
The group $\mathcal{O}$ can be generated by $\mathfrak{j}_{\pi/2}$, rotating $(\bm{m}_1,\bm{m}_2,\bm{m}_3)$ into $(\bm{m}_1,\bm{m}_3,-\bm{m}_2)$, together with $\mathfrak{b}_2$ and $\mathfrak{r}_3$. 
Since $\mathfrak{j}_{\pi/2}^2=\mathfrak{j}_{\pi}$, we can see that $\mathcal{T}\subseteq\mathcal{O}$, so $\mathbb{A}^{\mathcal{O},l}\subseteq \mathbb{A}^{\mathcal{T},l}$. 
Now let us consider the subgroup $\{\mathfrak{i},\mathfrak{j}_{\pi/2},\mathfrak{j}_{\pi/2}^2,\mathfrak{j}_{\pi/2}^3\}$.
For the first tensor in \eqref{tensors_poly} with odd $j_2$, and the second tensor in \eqref{tensors_poly} with even $j_2$, the rotation $\mathfrak{j}_{\pi/2}$ transforms them into their opposites, so they are vanishing. 
The remaining tensors in \eqref{tensors_poly} are invariant, thus
\begin{align}
\mathbb{A}^{\mathcal{O},l}=\textrm{span}\bigg\{
&\Big((\bm{m}_1^2\bm{m}_2^2+\bm{m}_2^2\bm{m}_3^2+\bm{m}_3^2\bm{m}_1^2)^{j_1}(\bm{m}_1\bm{m}_2\bm{m}_3)^{j_2}\Big)_0,\ 4j_1+3j_2=l,\ j_2\text{ even }; \nonumber\\
&\Big((\bm{m}_1^2-\bm{m}_2^2)(\bm{m}_2^2-\bm{m}_3^2)(\bm{m}_3^2-\bm{m}_1^2)(\bm{m}_1^2\bm{m}_2^2+\bm{m}_2^2\bm{m}_3^2+\bm{m}_3^2\bm{m}_1^2)^{j_1}\nonumber\\&\qquad(\bm{m}_1\bm{m}_2\bm{m}_3)^{j_2}\Big)_0, 
\quad 6+4j_1+3j_2=l,\ j_2\text{ odd }\bigg\}.\label{tensors_poly2}
\end{align}

\begin{theorem}
  The invariant tensors for $\mathcal{T}$ and $\mathcal{O}$ are given by \eqref{tensors_poly} and \eqref{tensors_poly2}, respectively. 
\end{theorem}

At last, we discuss the icosahedral group $\mathcal{I}$, which contains all the proper rotations of a regular icosahedron.
It can be generated by $\mathfrak{j}_{\pi}$, $\mathfrak{b}_2$, $\mathfrak{r}_3$, and 
\begin{align}
  &\mathfrak{v}_5=\frac{1}{2}\left(\begin{array}{ccc}
    \phi & -1 & \phi-1\\
    1 & \phi-1 & -\phi\\
    \phi-1 & \phi & 1
  \end{array}\right),\quad \phi=\frac{1+\sqrt{5}}{2}. 
\end{align}
Note that $\mathcal{T}\subseteq \mathcal{I}$, so we have $\mathbb{A}^{\mathcal{I},n}\subseteq \mathbb{A}^{\mathcal{T},l}$.
Furthermore, an invariant tensor shall be invariant under $\mathfrak{v}_5$. 
\begin{theorem}
  $\mathbb{A}^{\mathcal{I},n}=\{T(\mathfrak{p}): T(\mathfrak{p})\in\mathbb{A}^{\mathcal{T},l},\ T(\mathfrak{pv}_5)=T(\mathfrak{p}) \}$. 
\end{theorem}

We shall derive the expression for the lowest order invariant tensor for $\mathcal{I}$. 
The theorem tells us to choose from $\mathbb{A}^{\mathcal{T},l}$.
For $l=1,2$, the space is a zero space.
For $l=3$, the space is one-dimensional containing $\bm{m}_1\bm{m}_2\bm{m}_3$.
For $l=4$, the space is also one-dimensional, containing $(\bm{m}_1^2\bm{m}_2^2+\bm{m}_2^2\bm{m}_3^2+\bm{m}_3^2\bm{m}_1^2)_0$.
However, these two tensors are not invariant under $\mathfrak{v}_5$.
For $l=5$, the space is a zero space.
Therefore, we look into $l=6$. It is a two-dimensional space. 
We seek the tensor written in the form 
\begin{equation}
  H(\mathfrak{p})=a(\bm{m}_1^2\bm{m}_2^2\bm{m}_3^2)_0+b\Big((\bm{m}_1^2-\bm{m}_2^2)(\bm{m}_2^2-\bm{m}_3^2)(\bm{m}_3^2-\bm{m}_1^2)\Big)_0, 
\end{equation}
which satisfies $H(\mathfrak{pv}_5)=H(\mathfrak{p})$. 
Taking the expressions of \eqref{SymTrls_monomial} into the above, we can solve after a long calculation that $b=-\frac{\sqrt{5}}{11}a$. 
Thus, $\mathbb{A}^{\mathcal{I},6}$ is the lowest order nontrivial space. 

\subsection{Distinguishing groups by order parameters\label{orddis}}
In the above, we have written down the space of invariant tensors for all the point groups in $SO(3)$. 
For order parameters, we shall choose some from invariant tensors and average them about the density $\rho(\mathfrak{p})$. 
It might involve many considerations when making the choice of order parameters when studying a particular molecule. 
However, the choice shall anyway be able to distinguish one group from others.
To be specific, we notice from Proposition \ref{tensor_subgroup} that for two different groups $\mathcal{G}_1\subset\mathcal{G}_2$, we have $\mathbb{A}^{\mathcal{G}_2,l}\subseteq\mathbb{A}^{\mathcal{G}_1,l}$. 
When choosing order parameters for $\mathcal{G}_1$ from the invariant tensors, we need to choose at least one tensor from $(\mathbb{A}^{\mathcal{G}_2,l})^{\perp}\cap\mathbb{A}^{\mathcal{G}_1,l}$ to distinguish the two groups. 

We shall pay special attention for the distiguishing of the groups 
$\mathcal{D}_n$ and $\mathcal{D}_{\infty}$. 
It shall be noted that for $l<n$, we have $\mathbb{A}^{\mathcal{D}_n,l}=\mathbb{A}^{\mathcal{D}_{\infty},l}$. 
In order to distinguish these two groups, we need to at least include the $n$-th order tensor $\langle\tilde{T}^n(\bm{m}_2,\mathfrak{i}-\bm{m}_1^2)\rangle$. 
In other words, if we only choose some averaged tensors not greater than $(n-1)$-th order, we are not able to distinguish $\mathcal{D}_n$ from $\mathcal{D}_{\infty}$.
A similar requirement is suitable for distinguishing $\mathcal{C}_n$ and $\mathcal{C}_{\infty}$. 

The requirements above are based on the consideration of distinguishing molecular symmetries. 
The choice of order parameters also depends on our demand of classifying local anisotropy, which we will discuss in the next section. 


\section{Classifying local anisotropy\label{dist_sym}}
For the description of the local anisotropy formed by rigid molecules, a few averaged tensors are chosen as order parameters. 
Recall that for rod-like molecules, using the second order tensor $Q$, the local anisotropy is classfied by the eigenvalues. 
We discuss such a problem for general rigid molecules: suppose that we have chosen a few averaged tensors as order parameters, how to classify the local anisotropy from the values of these tensors? 
%
%
To be specific, we notice that for rod-like molecules, the isotropic, uniaxial and biaxial states are regarded to have different symmetries. 
Thus actually, we would like to ask under what values of the tensors does the local anisotropy have certain mesoscopic symmetry. 

Let us first state the symmetry of local anisotropy in mathematical language. 
Remember that the state in an infinitesimal volume is described by the density function $\rho(\mathfrak{p})$.
We rotate this infinitesimal volume by certain $\mathfrak{t}\in SO(3)$.
Under such a rotation, each body-fixed frame $\mathfrak{p}$ is transformed into $\mathfrak{tp}$. 
Thus, the resulting state is given by $\rho(\mathfrak{tp})$. 
We consider all the $\mathfrak{t}$ leaving the volume invariant, i.e.
\begin{equation}
  \rho(\mathfrak{tp})=\rho(\mathfrak{p}). \label{mesosym}
\end{equation}
They form a subgroup $\mathcal{H}$ of $SO(3)$.
The symmetry of local anisotropy is described by the group $\mathcal{H}$. 

However, we only know the values of some tensors averaged by $\rho(\mathfrak{p})$.
Generally, it is not sufficient to determine the density function. 
Of all the density functions giving these averaged tensors, we shall consider the one that maximizes the entropy. 
This approach has been done before for rod-like and bent-core molecules \cite{RodModel,BentModel}. 
Such a density function could be regarded as the equilibrium state for the system with the averaged tensors constrained by some forces. 

\subsection{From tensor to density function}
Assume that the symmetric traceless tensors we have chosen are $\langle U_1^{n_1}(\mathfrak{p})\rangle,\ldots,\langle U_l^{n_l}(\mathfrak{p})\rangle$. 
Hereafter, we will use the superscript $n_j$ to represent the order of the tensor $U_j$. 
To discuss the symmetry of the local anisotropy, let us consider the density function maximizing the entropy, or minimizing the following functional, 
\begin{equation}
\int\rho(\mathfrak{p})\ln\rho(\mathfrak{p})\md \mathfrak{p}, \label{entropy}
\end{equation}
under the constraints
\begin{align}
  \int\rho(\mathfrak{p})\,\md \mathfrak{p}=1,\quad
  \int\rho(\mathfrak{p})\,U_j^{n_j}(\mathfrak{p})\,\md \mathfrak{p}=W_j^{n_j}. \label{tensor_value}
\end{align}

\begin{lemma}\label{boltz}
  For any chosen tensors $\langle U_1^{n_1}(\mathfrak{p})\rangle,\ldots,\langle U_l^{n_l}(\mathfrak{p})\rangle$, assume that there exists a density function $0\le\rho(\mathfrak{p})< +\infty$ such that \eqref{tensor_value} holds. 
  Then, under the constraints \eqref{tensor_value}, there exists a unique density function $\rho(\mathfrak{p})$ that minimizes \eqref{entropy}.
  It is given in the following form, 
  \begin{align}
    \rho(\mathfrak{p})=\frac{1}{Z}\exp\bigg(\sum_{j=1}^l\sum_{s=1}^{2n_j+1}b_{js}U_j^{n_j}(\mathfrak{p})\cdot X^{n_j}_{s}\bigg), \label{form1}
  \end{align}
  where $X^{n_j}_1,\ldots,X^{n_j}_{2n_j+1}$ gives a basis of the $n_j$-th order symmetric traceless tensors, $b_{js_j}$ and $Z$ are constants to satisfy the constraints \eqref{tensor_value}. 
\end{lemma}
\begin{proof}
  For the uniqueness, we just notice that $\rho\ln\rho$ is strictly convex about $\rho$.

  We shall next write down the Euler--Langrange equation for when minimizing \eqref{entropy} under the constraints \eqref{tensor_value},  
  \begin{align*}
    Z+\ln\rho(\mathfrak{p})=\sum_{j,s}b_{js}U_j^{n_j}(\mathfrak{p})\cdot X^{n_j}_s. 
  \end{align*}
  It is equivalent to \eqref{form1}. 
  Thus, if there exists $b_{js}$ and $Z$ that solve the Euler--Lagrange equation and \eqref{tensor_value}, they give the unique solution to the constrainted minimization problem. 

  We show the existence below. Construct a function 
  \begin{align}
    J(b_{js})=\ln\int\exp\Big(\sum_{j,s}b_{js}\big(U_j^{n_j}(\mathfrak{p})- W_j^{n_j}\big)\cdot X^{n_j}_s\Big)\,\md \mathfrak{p}. 
  \end{align}
  We shall prove that the function has a stationary point.
  
  The condition in the lemma indicates that we could find a $0\le\rho_1(\mathfrak{p})<+\infty$ such that 
  \begin{align}
    \int \rho_1(\mathfrak{p})\sum_{j,s}b_{js}\Big(U_j^{n_j}(\mathfrak{p})-W_j^{n_j}\Big)\cdot X^{n_j}_s\,\md \mathfrak{p}=0. 
  \end{align}
  We shall prove that there exists some $\mathfrak{p}\in SO(3)$ such that 
  \begin{align}
    \sum_{j,s}b_{js}\Big(U_j^{n_j}(\mathfrak{p})-W_j^{n_j}\Big)\cdot X^{n_j}_s>0. \label{fun1}
  \end{align}
  Actually, if it does not hold, then it is nonpositive for any $\mathfrak{p}$.
  Because the average is zero, we must have $\rho_1(\mathfrak{p})=0$ if it is negative. By Proposition \ref{wt}, we know that the functions $\big(U_j^{n_j}(\mathfrak{p})-W_j^{n_j}\big)\cdot X^{n_j}_s$ are linearly independent. Thus, when $b_{js}$ are not all zero, the left-hand side of \eqref{fun1} is a nonzero function.
  Since the left-hand side of \eqref{fun1} could be written as a trigonometric polynomial using the Euler angles expression, we conclude that all the $\mathfrak{p}$ that make it zero forms a zero-measure set. 
  This contradicts $\rho_1<+\infty$. 

  Therefore, for any $b_{js}$ that are not all zero, we could find a $\mathfrak{p}$ such that \eqref{fun1} holds. Since the left-hand side of \eqref{fun1} is continuous on $SO(3)$, it follows that \eqref{fun1} holds in a neighborhood of $\mathfrak{p}$.
  We then deduce that $\lim_{a\to +\infty}J(ab_{js})=+\infty$. 
  Thus, the function $J$ has at least one stationary point. For any stationary point, we take the values $b_{js}$ into \eqref{form1} and calculate the normalizing constant $Z$ accordingly. It is straightforward to verify that the resulting $\rho(\mathfrak{p})$ satisfies \eqref{tensor_value}. 
\end{proof}

We turn to the symmetry of the density function $\rho(\mathfrak{p})$. 
If the density function satisfies $\rho(\mathfrak{tp})=\rho(\mathfrak{p})$ for any $\mathfrak{t}\in \mathcal{H}$, then we have 
\begin{align}
  \langle U_j^{n_j}(\mathfrak{p})\rangle=&\int U_j^{n_j}(\mathfrak{p})\rho(\mathfrak{p})\,\md \mathfrak{p}=\int U_j^{n_j}(\mathfrak{tp})\rho(\mathfrak{tp})\,\md \mathfrak{p}
  =\int \mathfrak{t}\circ U_j^{n_j}(\mathfrak{p})\rho(\mathfrak{p})\md \mathfrak{p}=\mathfrak{t}\circ \langle U_j^{n_j}(\mathfrak{p})\rangle, 
\end{align}
for any $\mathfrak{t}\in\mathcal{H}$. 
Therefore, using the statements in the previous section, we must have
\begin{align}
  \langle U_j^{n_j}(\mathfrak{p})\rangle=W_j^{n_j},\quad W_j^{n_j}(\mathfrak{p})\in\mathbb{A}^{\mathcal{H},n_j}. 
\end{align}

We discuss the inverse of the above statement.
\begin{theorem}\label{mesosym1}
  If there exists a $\mathfrak{q}\in SO(3)$ such that $\langle U_j^{n_j}(\mathfrak{p})\rangle=Z^{n_j}_j(\mathfrak{q})$, where $Z^{n_j}_j(\mathfrak{p})\in\mathbb{A}^{\mathcal{H},n_j}$. Then, the density function $\rho(\mathfrak{p})$ that maximizes the entropy, if it exists, satisfies $\rho(\mathfrak{qtq}^{-1}\mathfrak{p})=\rho(\mathfrak{p})$. 
\end{theorem}
\begin{proof}
  For each $n_j$, choose $\{V^{n_j}_s\}$ as a basis of $\mathbb{A}^{\mathcal{H},n_j}$. Let us consider the minimization problem below, 
  \begin{align}
    &\min\int \rho\ln\rho\,\md \mathfrak{p}, 
    \qquad\text{s.t. }\int\rho(\mathfrak{p})\,\md \mathfrak{p}=1,\quad
    \int\rho(\mathfrak{p})\,U_j^{n_j}(\mathfrak{p})\cdot V^{n_j}_s(\mathfrak{q})\,\md \mathfrak{p}=Z_j^{n_j}(\mathfrak{q})\cdot V^{n_j}_s(\mathfrak{q}). \label{tensym}
  \end{align}
  Following the same route of Lemma \ref{boltz}, if the solution exists, it must be unique and takes the form 
  \begin{align}
    \rho(\mathfrak{p})=\frac{1}{Z}\exp\Big(\sum_{j,s}b_{js}U_j^{n_j}(\mathfrak{p})\cdot V^{n_j}_s(\mathfrak{q})\Big). \label{rhosym}
  \end{align}
  By Proposition \ref{inv_van}, $V^{n_j}_s(\mathfrak{pt})=V^{n_j}_s(\mathfrak{p})$ for any $\mathfrak{t}\in\mathcal{H}$. 
  So we have 
  \begin{align*}
    \rho(\mathfrak{qtp})=&\frac{1}{Z}\exp\Big(\sum_{j,s}b_{js}U_j^{n_j}(\mathfrak{qtp})\cdot V^{n_j}_s(\mathfrak{q})\Big)
    =\frac{1}{Z}\exp\Big(\sum_{j,s} b_{js}U_j^{n_j}(\mathfrak{p})\cdot V^{n_j}_s(\mathfrak{t}^{-1})\Big)\\
    =&\frac{1}{Z}\exp\Big(\sum_{j,s}b_{js}U_j^{n_j}(\mathfrak{p})\cdot V^{n_j}_s(\mathfrak{i})\Big)
    =\frac{1}{Z}\exp\Big(\sum_{j,s}b_{js}U_j^{n_j}(\mathfrak{qp})\cdot V^{n_j}_s(\mathfrak{q})\Big)\\
    =&\rho(\mathfrak{qp}). 
  \end{align*}
  The remaining is to show that the above density indeed has the moments $\langle U_j^{n_j}(\mathfrak{p})\rangle=Z_j^{n_j}(\mathfrak{q})$. Since we have constraints \eqref{tensym} and $V_j^{n_j}(\mathfrak{p})$ give a basis of $\mathbb{A}^{\mathcal{H},n_j}$, we only need to show $\langle U_j^{n_j}(\mathfrak{p})\rangle\cdot Y_j(\mathfrak{q})=0$ for any $Y_j(\mathfrak{p})\in (\mathbb{A}^{\mathcal{H},n_j})^{\perp}$. Indeed, for any $\mathfrak{t}\in\mathcal{H}$, we have 
\begin{align*}
  \langle U_j^{n_j}(\mathfrak{p})\rangle\cdot Y_j(\mathfrak{q})=
  &\int U_j^{n_j}(\mathfrak{p})\cdot Y_j(\mathfrak{q})\rho(\mathfrak{p})\,\md\mathfrak{p}
  =\int U_j^{n_j}(\mathfrak{qtq}^{-1}\mathfrak{p})\cdot Y_j(\mathfrak{q})\rho(\mathfrak{qtq}^{-1}\mathfrak{p})\,\md\mathfrak{p}\nonumber\\
  =&\int U_j^{n_j}(\mathfrak{p})\cdot Y_j(\mathfrak{qt}^{-1})\rho(\mathfrak{qtq}^{-1}\mathfrak{p})\,\md\mathfrak{p}
  =\int U_j^{n_j}(\mathfrak{p})\cdot Y_j(\mathfrak{qt}^{-1})\rho(\mathfrak{p})\,\md\mathfrak{p}\nonumber\\
  =&\langle U_j^{n_j}(\mathfrak{p})\rangle\cdot Y_j(\mathfrak{qt}^{-1}). 
\end{align*}
Let us take the sum over $\mathfrak{t}\in\mathcal{H}$.
Using Proposition \ref{inv_van}, we have $Y_j^{\mathcal{H}}(\mathfrak{q})=0$, yielding 
\begin{align*}
  \sum_{\mathfrak{t}\in\mathcal{H}}\langle U_j^{n_j}(\mathfrak{p})\rangle\cdot Y_j(\mathfrak{qt}^{-1})=0, 
\end{align*}
which implies $\langle U_j^{n_j}(\mathfrak{p})\rangle\cdot Y_j(\mathfrak{q})=0$. 
\end{proof}

In the above theorem, we incorporate a rotation $\mathfrak{q}$ to take an appropriate choice of the reference frame into considertaion, 
which is analogous to diagonalizing the tensor $Q$ for the rod-like molecules. 
The theorem actually indicates that the maximum entropy solution gives the highest symmetry allowed by certain value of tensors. 



\subsection{Classification by tensors}
Theorem \ref{mesosym1} can be used to discuss the classification of local anisotropy. 
If all of $\langle U_j^{n_j}(\mathfrak{p})\rangle$ are zero, the infinitesimal volume is isotropic. 
By allowing some of these tensors to take nonzero values, the isotropic state will be broken into an anisotropic state. 
In general, let us consider two groups $\mathcal{H'}\subset\mathcal{H}$. 
By Proposition \ref{tensor_subgroup}, we have $\mathbb{A}^{\mathcal{H},n}\subseteq\mathbb{A}^{\mathcal{H'},n}$. 
If for some $n_j$, $\mathbb{A}^{\mathcal{H},n_j}$ is a proper subset of $\mathbb{A}^{\mathcal{H'},n_j}$, the averaged tensors $\langle U_j^{n_j}(\mathfrak{p})\rangle$ are allowed to take values in larger spaces, so that the mesoscopic symmetry $\mathcal{H}$ is broken into its subgroup $\mathcal{H'}$.
On the other hand, if for all $n_j$ we have $\mathbb{A}^{\mathcal{H},n_j}=\mathbb{A}^{\mathcal{H'},n_j}$, then $\mathcal{H'}$ cannot be distinguished from $\mathcal{H}$ by the tensors $\langle U_j^{n_j}(\mathfrak{p})\rangle$.
Therefore, the choice of order parameter tensors will affect the ability to classify mesoscopic symmetries. 

In what follows, we aim to discuss the mesoscopic symmetry breaking for some cases. 
We examine four molecular symmetries $\mathcal{D}_{\infty}$, $\mathcal{D}_2$, $\mathcal{C}_2$ and $\mathcal{T}$. 
For some of the molecular symmetries, different choices of order parameters will also be considered, which will lead to different classification. 
Consider the five tensors below: 
\begin{align*}
&Q^1=\langle\bm{m}_1\rangle,\
Q^2=\langle\bm{m}_1^2-\frac{1}{3}\mathfrak{i}\rangle,\
Q^4=\langle\bm{m}_1^4 - \frac{6}{7}\bm{m}_1^2\mathfrak{i} + \frac{3}{35}\mathfrak{i}^2\rangle,\\
&M^2=\langle\bm{m}_2^2-\frac{1}{2}(\mathfrak{i}-\bm{m}_1^2)\rangle,
T^3=\langle 2\bm{m}_1\bm{m}_2\bm{m}_3\rangle.
\end{align*}
Specifically, we discuss the six cases: 
\begin{itemize}
\item $\mathcal{D}_{\infty}$ molecular symmetry, with (1) one tensors $Q^2$; (2) two tensors $Q^2$, $Q^4$. 
\item $\mathcal{D}_2$ molecular symmetry, with (1) two tensors $Q^2$, $M^2$; (2) three tensors $Q^2$, $M^2$, $T^3$. 
\item $\mathcal{C}_2$ molecular symmetry, with three tensors $Q^1$, $Q^2$, $M^2$. 
\item $\mathcal{T}$ molecular symmetry, with one tensor $T^3$. 
\end{itemize}
Let us express the tensors by the basis of symmetric traceless tensors given in \eqref{SymTrls_gamma} with the frame $(\bm{e}_1,\bm{e}_2,\bm{e}_3)$ replaced by $\mathfrak{q}=(\bm{q}_1,\bm{q}_2,\bm{q}_3)$: 
\begin{subequations}\label{fivetensors}
\begin{align}
  Q^1=&d_1\bm{q}_1+d_2\bm{q}_2+d_3\bm{q}_3,\label{Q1}\\
  Q^2=&
  a_1(\bm{q}_1^2-\frac{1}{3}\mathfrak{i})+a_2\big(2\bm{q}_2^2-(\mathfrak{i}-\bm{q}_1^2)\big)+2a_3\bm{q}_2\bm{q}_3+a_4\bm{q}_1\bm{q}_3+a_5\bm{q}_2\bm{q}_3,\label{Q2}\\
  M^2=&
  a'_1(\bm{q}_1^2-\frac{1}{3}\mathfrak{i})+a'_2\big(2\bm{q}_2^2-(\mathfrak{i}-\bm{q}_1^2)\big)+2a'_3\bm{q}_2\bm{q}_3+a'_4\bm{q}_1\bm{q}_3+a'_5\bm{q}_2\bm{q}_3,\label{M2}\\
  T^3=&b_1(\bm{q}_1^3-\frac{3}{5}\mathfrak{i}\bm{q}_1)
  +b_2\big(4\bm{q}_2^3-3(\mathfrak{i}-\bm{q}_1^2)\bm{q}_2\big)
  +b_3\big(4\bm{q}_2^2\bm{q}_3-(\mathfrak{i}-\bm{q}_1^2)\bm{q}_3\big)\nonumber\\
  &+2b_4\bm{q}_1\bm{q}_2\bm{q}_3+b_5\bm{q}_1\big(2\bm{q}_2^2-(\mathfrak{i}-\bm{q}_1^2)\big)
  +b_6(\bm{q}_1^2-\frac{1}{5}\mathfrak{i})\bm{q}_2+b_7(\bm{q}_1^2-\frac{1}{5}\mathfrak{i})\bm{q}_3,\label{T3}\\
  Q^4=&
  c_1(\bm{q}_1^4 - \frac{6}{7}\mathfrak{i}\bm{q}_1^2 + \frac{3}{35}\mathfrak{i}^2)+c_2(\bm{q}_1^2-\frac{1}{7}\mathfrak{i})\Big(2\bm{q}_2^2-(\mathfrak{i}-\bm{q}_1^2)\Big)+c_3(\bm{q}_1^2-\frac{1}{7}\mathfrak{i})2\bm{q}_2\bm{q}_3\nonumber\\
  &+c_4\big(8\bm{q}_2^4-8(\mathfrak{i}-\bm{q}_1^2)\bm{q}_2^2+(\mathfrak{i}-\bm{q}_1^2)^2\big)
  +c_5\big(8\bm{q}_2^3\bm{q}_3-4(\mathfrak{i}-\bm{q}_1^2)\bm{q}_2\bm{q}_3\big)\nonumber\\
  &+c_6\bm{q}_1\big(4\bm{q}_2^2\bm{q}_3-(\mathfrak{i}-\bm{q}_1^2)\bm{q}_3\big)
  +c_7\bm{q}_1\big(4\bm{q}_2^3-3(\mathfrak{i}-\bm{q}_1^2)\bm{q}_2\big)\nonumber\\
  &+c_8(\bm{q}_1^3-\frac{3}{7}\mathfrak{i}\bm{q}_1)\bm{q}_2+c_9(\bm{q}_1^3-\frac{3}{7}\mathfrak{i}\bm{q}_1)\bm{q}_3. \label{Q4}
\end{align}
\end{subequations}
When discussing mesoscopic symmetries, Theorem \ref{mesosym1} will impose conditions on the coefficients $d_i$, $a_i$, $a'_i$, $b_i$, $c_i$. 

\begin{table}
\centering
\begin{tabular}{|c|c|}\hline
  Mesoscopic & Nonzero coefficients allowed, \\
  symmetry & with constraints required\\\hline
 $\mathcal{D}_{\infty}$ & $a_1,a_1',c_1$ \\\hline
 $\mathcal{C}_{\infty}$ & $d_1,a_1,a_1',b_1,c_1$ \\\hline
  $\mathcal{O}$ & 1) $c_1=7c_4$; \\
  & or 2)  $-4\sqrt{2}c_1=7c_6$\\\hline
  $\mathcal{T}$ & 1) $b_4$, $c_1=7c_4$;\\
  & or 2) $\sqrt{2}b_1=5b_3$, $-4\sqrt{2}c_1=7c_6$\\\hline
 $\mathcal{D}_4$ & $a_1,a_1',c_1,c_4$ \\\hline
 $\mathcal{D}_3$ & $a_1,a_1',b_2,c_1,c_6$ \\\hline
 $\mathcal{D}_2$ & $a_1,a_2,c_1,c_2,c_4$ \\\hline
 $\mathcal{C}_4$ & $d_1,a_1,a_1',b_1,c_1,c_4,c_5$\\\hline
 $\mathcal{C}_3$ & $d_1,a_1,a_1',b_1,b_2,b_3,c_1,c_6,c_7$ \\\hline
 $\mathcal{C}_2$ & $d_1, a_1,a_2,a_3, a_1',a_2',a_3', b_1,b_4,b_5, c_1,c_2,c_3,c_4,c_5$ \\\hline
\end{tabular}
\caption{Nonzero coefficients allowed for mesoscopic symmetries, where for $\mathcal{O}$ and $\mathcal{T}$ some constraints are required. Also, for $\mathcal{O}$ and $\mathcal{T}$ two different conditions are listed under different choices of $\mathfrak{q}$, which are expalined in the text.\label{coefcond}}
\end{table}

For the mesoscopic symmetries, from the discussion in Section \ref{orddis} we have $\mathbb{A}^{\mathcal{C}_n,l}=\mathbb{A}^{\mathcal{C}_{\infty},l}$ and $\mathbb{A}^{\mathcal{D}_n,l}=\mathbb{A}^{\mathcal{D}_{\infty},l}$ for $l<n$. 
Since the tensors we consider are not greater than fourth order, we do not need to discuss $\mathcal{C}_n$  and $\mathcal{D}_n$ for $n\ge 5$, because
with the order parameters above, we are not able to distinguish the two mesoscopic symmetries $\mathcal{C}_n$ (resp. $\mathcal{D}_n$) and $\mathcal{C}_{\infty}$ (resp. $\mathcal{D}_{\infty}$). 
For the same reason, we do not need to discuss $\mathcal{I}$. 
The remaining groups are listed in the Table \ref{coefcond}.

For each group in Table \ref{coefcond}, we write down the conditions on the coefficients according to Theorem \ref{mesosym1}, with the $\mathfrak{q}$ in \eqref{fivetensors} identical to the $\mathfrak{q}$ in Theorem \ref{mesosym1}.
As a result, the condition in Theorem \ref{mesosym1} allows some coefficients to be nonzero, sometimes with constraints. 
Using the expression of invariant tensors in the previous section, it is straightforward to write down nonzero coefficients allowed for most groups, which are also listed in Table \ref{coefcond}. 
Let us explain the conditions for the two mesoscopic symmetries $\mathcal{O}$ and $\mathcal{T}$. 
Theorem \ref{mesosym1} requires $Q^4=\lambda(\bm{q}_1^2\bm{q}_2^2+\bm{q}_2^2\bm{q}_3^2+\bm{q}_3^2\bm{q}_1^2)_0$ for some constant $\lambda$. Using \eqref{SymTrls_monomial2}, we deduce that 
\begin{align*}
  (\bm{q}_1^2\bm{q}_2^2+\bm{q}_2^2\bm{q}_3^2+\bm{q}_3^2\bm{q}_1^2)_0=&(\bm{q}_1^2\bm{q}_2^2+(\bm{q}_1^2+\bm{q}_2^2)(\mathfrak{i}-\bm{q}_1^2-\bm{q}_2^2))_0
  =-(\bm{q}_1^4+\bm{q}_1^2\bm{q}_2^2+\bm{q}_2^4)_0\\
  =&-\frac{1}{8}\big(8\bm{q}_2^4-8(\mathfrak{i}-\bm{q}_1^2)\bm{q}_2^2+(\mathfrak{i}-\bm{q}_1^2)^2\big)
  -\frac{7}{8}(\bm{q}_1^4 - \frac{6}{7}\bm{q}_1^2\mathfrak{i} + \frac{3}{35}\mathfrak{i}^2). 
\end{align*}
Therefore, we obtain the condition $c_1=7c_4$. 
Another thing to be explained is that we write down two different conditions for $\mathcal{O}$ and $\mathcal{T}$.
They are actually equivalent, but shall be understood in different choice of the frame $\mathfrak{q}$.
The conditions 2) for $\mathcal{O}$ and $\mathcal{T}$ are actually derived in the frame 
\begin{align*}
  \mathfrak{q'}=(\bm{q'}_1,\bm{q'}_2,\bm{q'}_3)=(\bm{q}_1,\bm{q}_2,\bm{q}_3)
  \left(
  \begin{array}{ccc}
    \frac{1}{\sqrt{3}} & \frac{1}{\sqrt{2}} & \frac{1}{\sqrt{6}}\\
    \frac{1}{\sqrt{3}} & -\frac{1}{\sqrt{2}} & \frac{1}{\sqrt{6}}\\
    \frac{1}{\sqrt{3}} & 0 & -\frac{2}{\sqrt{6}}
  \end{array}
  \right).
\end{align*}
With some direct but tedious calculations using \eqref{SymTrls_monomial}--\eqref{SymTrls_monomial2}
, we obtain 
\begin{align*}
  &\bm{q}_1\bm{q}_2\bm{q}_3=\frac{5}{6\sqrt{3}}(\bm{q'}_1{}^3)_0+\frac{1}{3\sqrt{6}}(4\bm{q'}_2{}^2-(\mathfrak{i}-\bm{q'}_1{}^2))\bm{q'}_3. \\
  &(\bm{q}_1^2\bm{q}_2^2+\bm{q}_2^2\bm{q}_3^2+\bm{q}_3^2\bm{q}_1^2)_0
  =-\frac{7}{12}(\bm{q'}_1{}^4)_0+\frac{\sqrt{2}}{3}\bm{q'}_1(4\bm{q'}_2{}^2-(\mathfrak{i}-\bm{q'}_1{}^2))\bm{q'}_3. 
\end{align*}
In the frame $\mathfrak{q'}$, we actually put $\bm{q'}_1$ as the three-fold axis for $\mathcal{O}$ and $\mathcal{T}$. 
The conditions 2) follow from the right-hand side of the two equations above. 

The conditions listed in Table \ref{coefcond} are for the coefficients of all the five tensors.
One can notice that the conditions on coefficients for different groups are distinct if all the five tensors are chosen as order parameters. 
However, for the six cases we will discuss, we are including a part of the five tensors. 
So, for each of the six cases, only the conditions for chosen tensors in the second column of Table \ref{coefcond} will be effective.
As a result, not all the groups in Table \ref{coefcond} can be recognized in classification. 

Because we will discuss the breaking of one symmetry group into its subgroups, 
let us write down the subgroup relations below. 
They can be verified by the generating elements discussed in Section \ref{mol_sym}, 
$$
\mathcal{C}_n\subseteq\mathcal{D}_n\subseteq\mathcal{D}_{\infty},\
\mathcal{C}_n\subseteq\mathcal{C}_{\infty}\subseteq\mathcal{D}_{\infty},\ 
\mathcal{C}_2\subseteq\mathcal{C}_4,\ 
\mathcal{D}_2\subseteq\mathcal{D}_4,\ 
\mathcal{D}_2\subseteq\mathcal{T}\subseteq\mathcal{O},\
\mathcal{D}_2\subseteq\mathcal{D}_4\subseteq\mathcal{O}. 
$$
The following two relations need to be comprehended in the frame $\mathfrak{q'}$ for $\mathcal{T}$ and $\mathcal{O}$. 
$$
\mathcal{C}_3\subseteq\mathcal{T},\
\mathcal{D}_3\subseteq\mathcal{O}.
$$

\subsubsection{$\mathcal{D}_{\infty}$ molecular symmetry, one tensor $Q^2$}
We use this well-understood case to illustrate how we arrive at the classification. 
Since only one the tensor $Q^2$ is the order parameter, in Table \ref{coefcond} only the conditions on $a_i$ are effective. 
From the isotropic state, by allowing $a_1$ to be nonzero one obtains $\mathcal{D}_{\infty}$ mesoscopic symmetry, which is the uniaxial state. 
Further allowing $a_2$ to be nonzero one obtains $\mathcal{D}_{2}$, which is the biaxial state. 
We draw a graph showing the above relations in Fig. \ref{meso1} (left). 
In the graph, we write down the group and the nonzero coefficients allowed.
Arrows are drawn from a group to its subgroup. One could compare the nonzero coefficients allowed to find out what are the additional nonzero coefficients for certain symmetry breaking. 

As for other groups, for $\mathcal{C}_3$, $\mathcal{C}_4$, $\mathcal{C}_{\infty}$, $\mathcal{D}_3$, and $\mathcal{D}_4$, the only nonzero coefficient allowed is $a_1$. 
Thus, we cannot distinguish them from $\mathcal{D}_{\infty}$. 
For $\mathcal{T}$ and $\mathcal{O}$, it requires $Q^2=0$. 
Therefore, they also do not appear in the graph. 
Let us explain why $\mathcal{C}_2$ also does not appear. 
Under $\mathcal{C}_2$ mesoscopic symmetry, the coefficients $a_1,a_2,a_3$ could be nonzero. 
Consider another frame 
$$
\tilde{\mathfrak{q}}=(\tilde{\bm{q}}_1, \tilde{\bm{q}}_2, \tilde{\bm{q}}_3)=(\bm{q}_1,\bm{q}_2,\bm{q}_3)\mathfrak{j}_{\theta},
$$
where we recall that $\mathfrak{j}_{\theta}$ is defined in \eqref{basicrot}. 
We have $\tilde{\bm{q}}_1=\bm{q}_1$ and  
$(\tilde{\bm{q}}_2+\sqrt{-1}\tilde{\bm{q}}_3)^k=e^{-\sqrt{-1}\cdot k\theta}(\bm{q}_2+\sqrt{-1}\bm{q}_3)^k$ using \eqref{m23rotate}. 
Together with \eqref{chebyshev0}, we deduce that 
\begin{align}
&\cos k\theta\cdot\tilde{T}_k(\bm{q}_2,\mathfrak{i}-\bm{q}_1^2)-\sin k\theta\cdot\tilde{U}_{k-1}(\bm{q}_2,\mathfrak{i}-\bm{q}_1^2)\bm{q}_3=\tilde{T}_k(\tilde{\bm{q}}_2,\mathfrak{i}-\tilde{\bm{q}}_1^2), \nonumber\\
&\sin k\theta\cdot\tilde{T}_k(\bm{q}_2,\mathfrak{i}-\bm{q}_1^2)+\cos k\theta\cdot\tilde{U}_{k-1}(\bm{q}_2,\mathfrak{i}-\bm{q}_1^2)\bm{q}_3=\tilde{U}_{k-1}(\tilde{\bm{q}}_2,\mathfrak{i}-\tilde{\bm{q}}_1^2)\tilde{\bm{q}}_3.\label{m23rotate2}
\end{align}
Let $k=2$, and choose a $\theta$ such that $a_2\sin 2\theta+a_3\cos 2\theta=0$,
to arrive at 
\begin{align*}
  &a_1(\bm{q}_1^2-\frac{1}{3}\mathfrak{i})+a_2\big(2\bm{q}_2^2-(\mathfrak{i}-\bm{q}_1^2)\big)+2a_3\bm{q}_2\bm{q}_3
  =a_1(\tilde{\bm{q}}_1^2-\frac{1}{3}\mathfrak{i})+\sqrt{a_2^2+a_3^2}\big(2\tilde{\bm{q}}_2^2-(\mathfrak{i}-\tilde{\bm{q}}_1^2)\big). 
\end{align*}
The above process to eliminate a nonzero coefficient by rotating the frame $\mathfrak{q}$ can be viewed as a special case of diagonalizing the tensor $Q$. 

\begin{figure}
  \centering
  \includegraphics[width=0.1\textwidth,keepaspectratio]{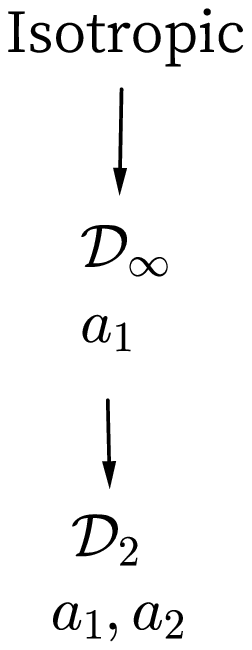}\hspace{36pt}
  \includegraphics[width=0.4\textwidth,keepaspectratio]{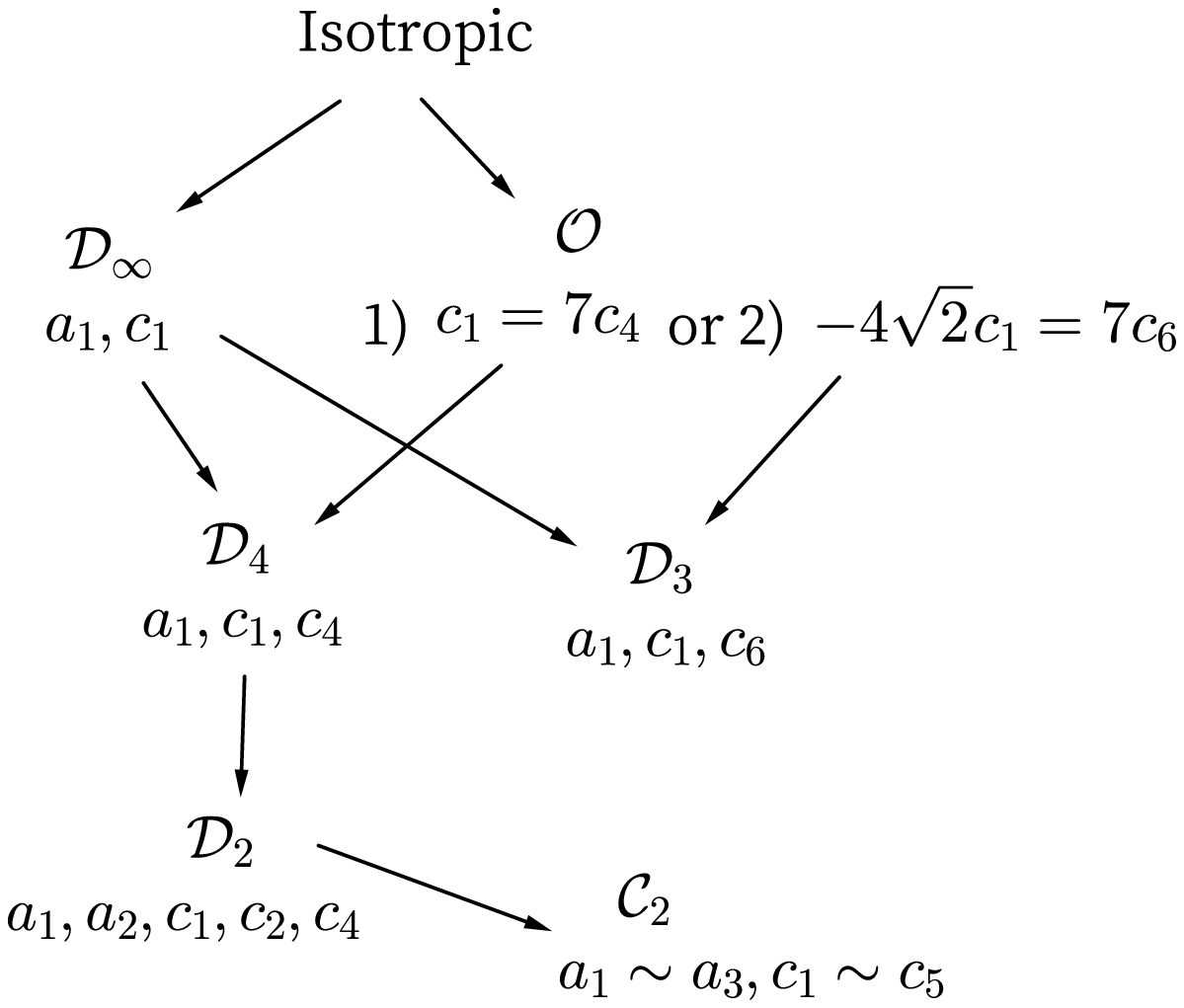}
  \caption{Classification of mesoscopic symmetries for $\mathcal{D}_{\infty}$ molecular symmetry, using one tensor $Q^2$ (left); two tensors $Q^2$, $Q^4$ (right). }\label{meso1}
\end{figure}
\begin{figure}
  \centering
  \includegraphics[width=0.18\textwidth,keepaspectratio]{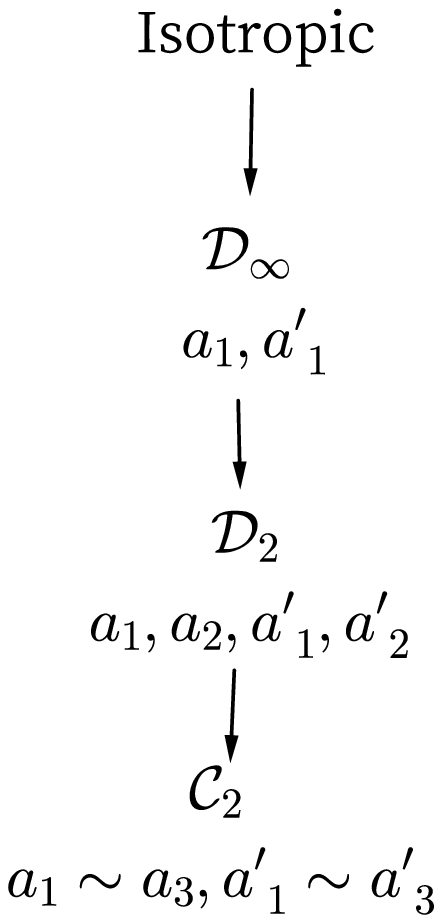}\hspace{36pt}
  \includegraphics[width=0.42\textwidth,keepaspectratio]{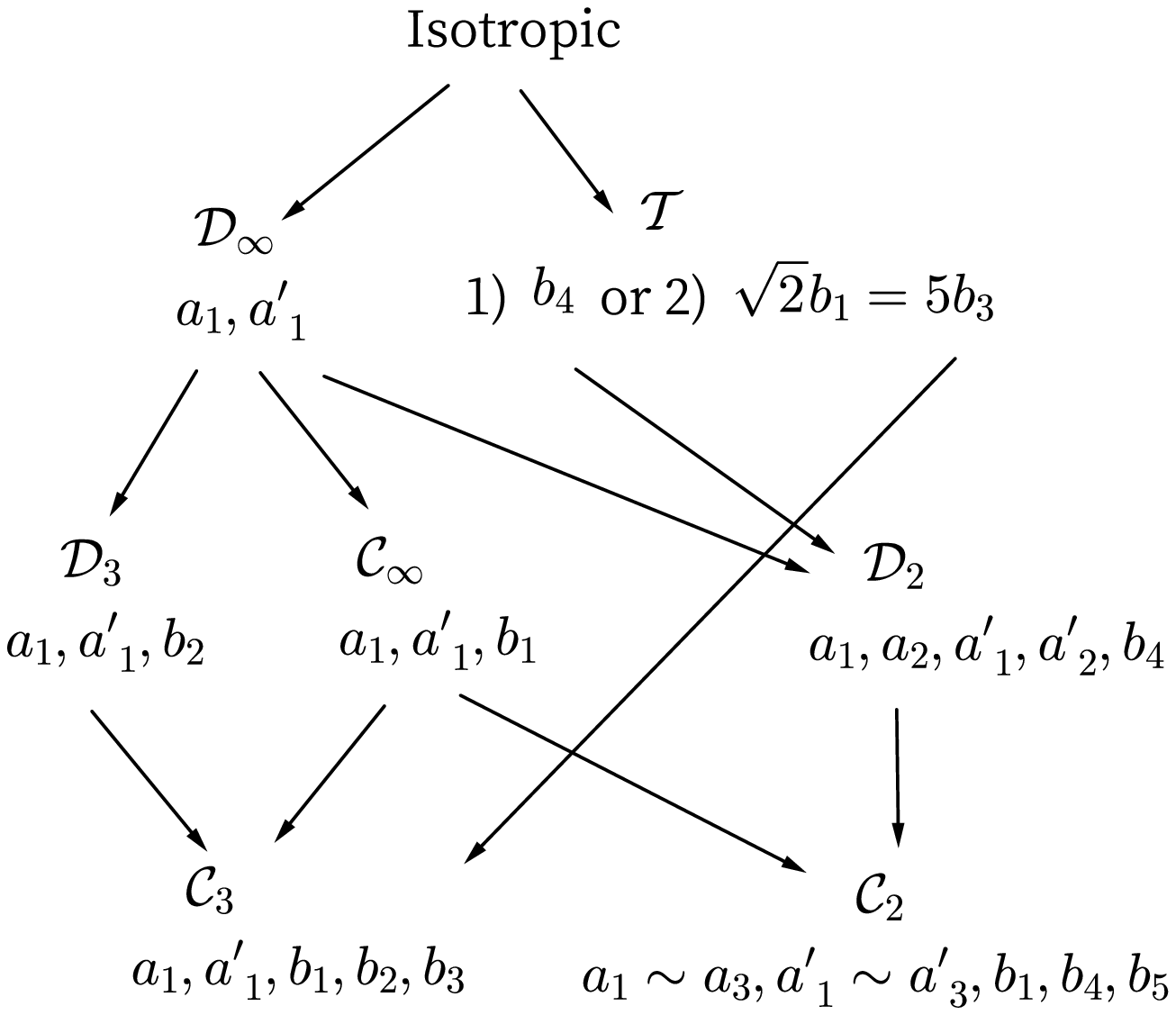}
  \caption{Classification of mesoscopic symmetries for $\mathcal{D}_2$ molecular symmetry, using two tensors $Q^2$, $M^2$ (left); three tensors $Q^2$, $M^2$, $T^3$ (right). }\label{meso2}
\end{figure}
\begin{figure}
  \centering
  \includegraphics[width=0.22\textwidth,keepaspectratio]{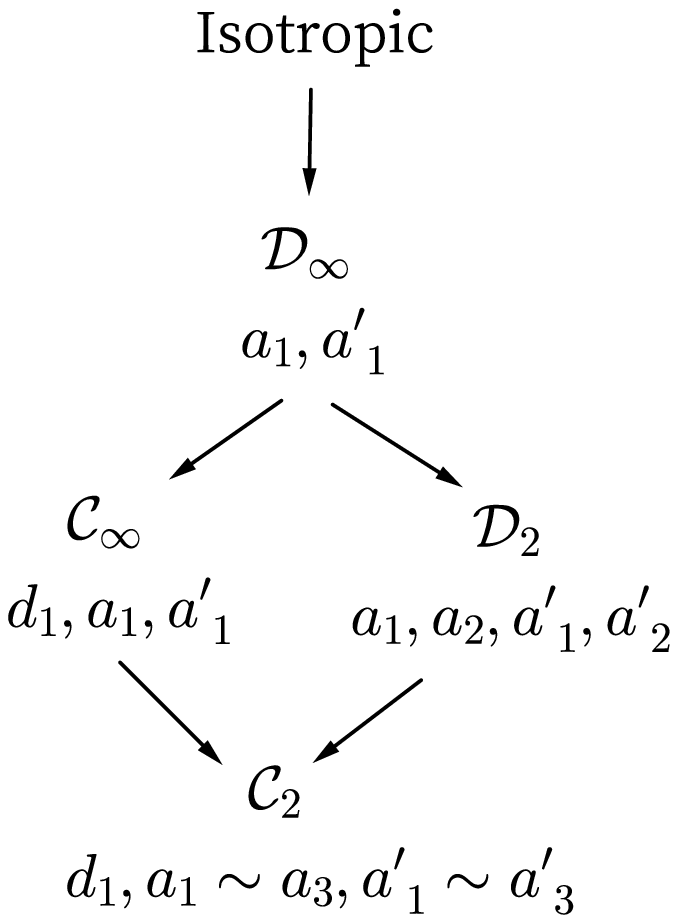}
  \includegraphics[width=0.41\textwidth,keepaspectratio]{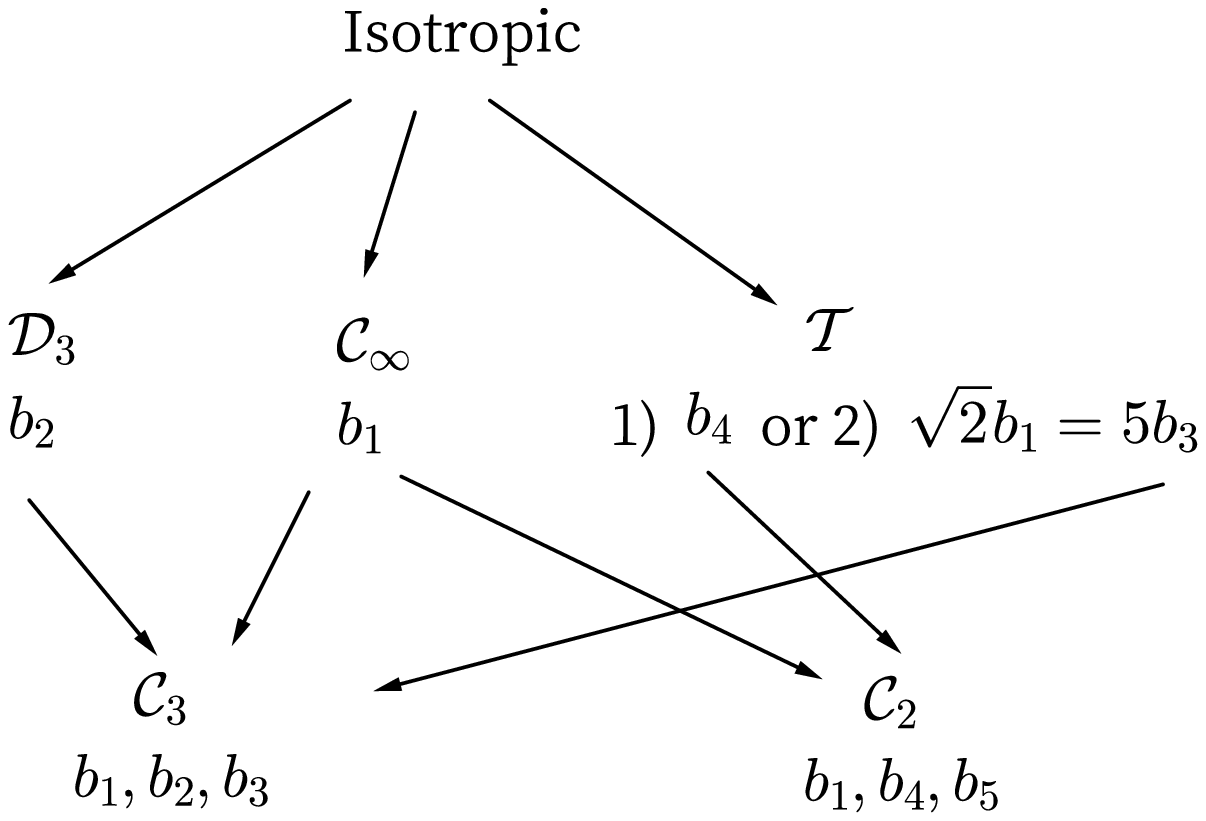}
  \caption{Classification of mesoscopic symmetries for $\mathcal{C}_2$ molecular symmetry with three tensors $Q^1$, $Q^2$, $M^2$ (left); for $\mathcal{T}$ molecular symmetry with one tensor $T^3$ (right). }\label{meso3}
\end{figure}

\subsubsection{$\mathcal{D}_{\infty}$ molecular symmetry, two tensors $Q^2$, $Q^4$}
We still examine the $\mathcal{D}_{\infty}$ molecular symmetry, but include two order parameter tensors $Q^2$ and $Q^4$. 
The graph for mesoscopic symmetries is given in Fig. \ref{meso1} (right). 
From the isotropic state, one could obtain $\mathcal{D}_{\infty}$ by allowing $a_1$ and $c_1$ nonzero; or, in another branch, by allowing either 1) $c_1=7c_2$ nonzero, or 2) $-4\sqrt{2}c_1=7c_6$ nonzero, to get $\mathcal{O}$. 
Then, from $\mathcal{D}_{\infty}$, one allows $c_4$ nonzero to obtain $\mathcal{D}_4$; $c_6$ nonzero to obtain $\mathcal{D}_3$. 
From $\mathcal{O}$, one allows $a_1$ nonzero and discards the constraint $c_1=7c_2$ of 1) to get $\mathcal{D}_4$. 
Still from $\mathcal{O}$, one discards the constraint $-4\sqrt{2}c_1=7c_6$ of 2) to get $\mathcal{D}_3$. 
Finally, starting from $\mathcal{D}_4$, by allowing $a_2,c_2$ nonzero one arrives at $\mathcal{D}_2$, further allowing $a_3,c_3,c_5$ nonzero to get $\mathcal{C}_2$. 

When using the tensors $Q^2$ and $Q^4$, the conditions for $\mathcal{C}_{\infty}$ are identical to those for $\mathcal{D}_{\infty}$, and the conditions for $\mathcal{T}$ are identical to those for $\mathcal{O}$. 
Thus, $\mathcal{C}_{\infty}$ and $\mathcal{T}$ do not appear in the graph. 
For $\mathcal{C}_3$, it allows $a_1,c_1,c_6,c_7$ nonzero.
One could use \eqref{m23rotate2} to determine a rotation making $c_7=0$, so that $\mathcal{C}_3$ cannot be distinguished from $\mathcal{D}_3$. 
For the same reason, $\mathcal{C}_4$ does not appear in the graph. 
Notice that $\mathcal{C}_2$ does appear in the graph. 
The difference between $\mathcal{C}_3$, $\mathcal{C}_4$ and $\mathcal{C}_2$ is that the rotation $\mathfrak{j}_{\theta}$ affects the coefficients $a_2,a_3,c_2,c_3,c_4,c_5$, which are all allowed nonzero by $\mathcal{C}_2$. Using \eqref{m23rotate2}, one could only make one of them zero, so that it cannot be reduced to a higher symmetry. 

Notice that when we include two tensors $Q^2$ and $Q^4$ as order parameters, they may describe local anisotropy with no symmetry. 
It will also happen for the remaining four cases to be presented. 
For the local anisotropy without any symmetry, we will give further discussion by the end of this section. 

\subsubsection{$\mathcal{D}_2$ molecular symmetry, two tensors $Q^2$, $M^2$}
We turn to the $\mathcal{D}_2$ molecular symmetry, and choose the two lowest order tensors $Q^2$ and $M^2$ as order parameters. 
The graph of mesoscopic symmetries is given in Fig. \ref{meso2} (left). 
From the isotropic state, one allows $a_1,a_1'$ nonzero to obtain $\mathcal{D}_{\infty}$, 
then allows $a_2,a_2'$ nonzero to obtain $\mathcal{D}_2$, 
finally allows $a_3,a_3'$ nonzero to obtain $\mathcal{C}_2$. 
Unlike the case of one second order tensor $Q^2$, we cannot always find a rotation to make $a_3$ and $a_3'$ vanish simultaneously. 


\subsubsection{$\mathcal{D}_2$ molecular symmetry, three tensors $Q^2$, $M^2$, $T^3$}
We still consider the $\mathcal{D}_2$ molecular symmetry, but include in the order parameters the third order tensor $T^3$ in addition to the two second order tensors. 
The graph of mesoscopic symmetries is given in Fig. \ref{meso2} (right). 
There are no $\mathcal{C}_4$, $\mathcal{D}_4$ and $\mathcal{O}$ because we do not have fourth order tensors in the order parameters. 
All the remaining seven groups appear in the graph. 
We do not explain the graph in detail, but point out two things. 
One is that the appearance of $\mathcal{C}_{\infty}$ and $\mathcal{C}_3$ is due to $b_1$. The other is that the connection from $\mathcal{T}$ to $\mathcal{C}_3$ can be realized by using the condition 2) in Table \ref{coefcond}: one does not require $\sqrt{2}b_1=5b_3$ and allows $a_1,a_1',b_1$ nonzero. 

\subsubsection{$\mathcal{C}_2$ molecular symmetry, three tensors $Q^1$, $Q^2$, $M^2$}
The group $\mathcal{C}_2$ gives the proper rotations allowed by a bent-core molecule. 
Here, the choice of three tensors is proposed in \cite{SymmO, BentModel}. 
The graph of mesoscopic symmetries is given in Fig. \ref{meso3} (left),  
where we could find four groups: $\mathcal{C}_{\infty}$, $\mathcal{D}_{\infty}$, $\mathcal{C}_2$, $\mathcal{D}_2$.
Compared with Fig. \ref{meso2} (left) for two tensors $Q^2$ and $M^2$, the group $\mathcal{C}_{\infty}$ appears, because we have a first order tensor and $d_1$ may become nonzero. 





\subsubsection{$\mathcal{T}$ molecular symmetry, one tensor $T^3$}
We finally discuss the tetrahedral molecular symmetry, and use the lowest order tensor $T^3$ as the order parameter. 
The graph of mesoscopic symmetries is given in Fig. \ref{meso3} (right). 
Let us compare with Fig. \ref{meso2} (right) for $\mathcal{D}_2$ molecular symmetry where $T^3$ is also one of the order parameters. 
We do not find $\mathcal{D}_{\infty}$, because it requires $T^3=0$. 
Besides, there is no $\mathcal{D}_2$, because the only nonzero coefficient it allows is $b_4$, but this is identical to the nonzero coefficient allowed by $\mathcal{T}$. 











\subsection{Discussions}
The main results in this section are shown by graphs connecting mesoscopic symmetries. 
From such a graph, one could figure out the ability to classify mesoscopic symmetry by certain set of tensors. 
To study when certain mesoscopic symmetry could occur, it 
would require detailed analysis, as is done in previous works \cite{fatkullin2005critical,ji2006study,liu2005axial,xu2017transmission}.

We discussed the conditions for mesoscopic symmetries allowing three- and four- fold rotations proposed from experiments \cite{JJAP}. 
In the theory of liquid crystals, usually only the lowest order averaged invariant tensors are kept as order parameters, such as $Q^2$ for rod-like molecules, and $Q^2,M^2$ for molecules with $\mathcal{D}_2$ symmetry. 
However, if there is evidence that three- or higher-fold rotations are allowed mesoscopically, one needs to include higher order tensors to adequately describe such states. 
In this sense, the choice of order parameters would be dependent on our demand on classifying mesoscopic symmetries. 
When trying to determine up to which order of tensors we need to keep, we shall consult the graphs we have given in this section. 

This is also the case for interpreting results of molecular simulations. 
Generally, one could obtain from molecular simulations an orientation distribution. 
In order to explain the result, one also needs to calculate some tensors averaged over this distribution. 
Which tensors should be calculated then becomes a significant problem.
It certainly depends on the molecular symmetry.
Meanwhile, one also needs to consider the demand of classifying local anisotropy, for which the graphs of mesoscopic symmetries are also useful. 

Under mesoscopic symmetry, the local anisotropy could be further classfied. 
For example, the uniaxial state for $Q^2$, which requires $Q^2=a_1(\bm{q}_1^2-\mathfrak{i}/3)$, is further classified by the sign of $a_1$. 
Nevertheless, the classification by the symmetries is generally the first level. 

We mentioned that the local anisotropy described by some tensors might have no symmetry. 
In this case, however, we can consider measuring the distance to certain mesoscopic symmetry. 
Suppose the order parameter tensors are $\langle U_j^{n_j}(\mathfrak{p})\rangle$ where $1\le j\le l$. The distance to a group $\mathcal{H}$ can be defined by 
\begin{align}
  \min_{\substack{\mathfrak{q}\in SO(3)\\W_j^{n_j}(\mathfrak{p})\in\mathbb{A}^{\mathcal{H},n_j}}}\sum_{j=1}^l|\langle U_j^{n_j}(\mathfrak{p})\rangle-W_j^{n_j}(\mathfrak{q})|^2. 
\end{align}
One could compare the distances to all the possible $\mathcal{H}$ to find which is the closest.

\section{Conclusion\label{concl}}
We discuss the description and classification of local anisotropy formed by rigid molecules in an infinitesimal volume, which is a fundamental problem in liquid crystals. 
With the consideration of identifying independent components, the order parameters shall be chosen from symmetric traceless tensors averaged by the density function. 
For certain molecular symmetry described by a point group in $SO(3)$, we shall eliminate the vanishing tensors and keep only the invariant tensors under this point group. 
For each point group in $SO(3)$, we write down the space of invariant tensors by explicit expressions. 
Once we have chosen some order parameter tensors according to the above principle, we could then classify the local anisotropy by its symmetry, i.e. the mesoscopic symmetry. 
By considering the maximum entropy state, the mesoscopic symmetry is determined by the value of order parameter tensors.
The conditions are also closely related to the space of invariant tensors. 
We discuss the classification for several sets of tensors, where three-fold, four-fold and polyhedral mesoscopic symmetries are included. 

Our results also provide information for the interpretation of results from molecular simulations. 
In a forthcoming work, we will utilize the results in the current paper to discuss the derivation of free energy about tensors from the molecular theory.

\appendix
\section{Quaternions}
We briefly describe how to use quaternions to express rotations.
A quaternion can be written as $\bm{q}=a+b\bm{i}+c\bm{j}+d\bm{k}$. 
The multiplication of quaternion follows $\bm{i}^2=\bm{j}^2=\bm{k}^2=-1,\ \bm{ij}=-\bm{ji}=\bm{k},\ \bm{jk}=-\bm{kj}, \bm{ki}=-\bm{ik}=\bm{j}$. 
Every unit quaternion with $a^2+b^2+c^2+d^2=1$ gives an element in $SO(3)$. 
For a vector $\bm{v}=(x,y,z)^T$, write it as $x\bm{i}+y\bm{j}+d\bm{k}$. 
The rotation is defined by $\bm{v}\mapsto \bm{q}(x\bm{i}+y\bm{j}+d\bm{k})\bm{q}^{-1}$. 
Obviously, $\bm{q}$ and $-\bm{q}$ yield the same rotation. 
The above definition actually gives the rotation matrix 
\begin{equation}
  \left(
  \begin{array}{ccc}
    a^2+b^2-c^2-d^2 & 2(bc-ad) & 2(ac+bd) \\
    2(ad+bc) & a^2-b^2+c^2-d^2 & 2(cd-ab) \\
    2(bd-ac) & 2(ab+cd) & a^2-b^2-c^2+d^2 
  \end{array}
  \right)
\end{equation}
The components are given by second order homogeneous polynomials about $(a,b,c,d)$.
A second order homogeneous polynomial about four variables has ten terms.
Eliminating $a^2+b^2+c^2+d^2=1$, there are nine terms remaining. 
Note that $\mathfrak{p}_{ij}$ has nine components. 
Thus, second order homogeneous polynomials about $(a,b,c,d)$ are linearly equivalent to $\mathfrak{p}_{ij}$.

\small
\bibliographystyle{plain}
\bibliography{bib_sym}

\end{document}